\documentclass[runningheads]{llncs}

\usepackage{xcolor}
\usepackage{stmaryrd}
\usepackage{amssymb}
\usepackage{amsmath}
\usepackage{graphicx}
\usepackage{float}
\usepackage{hyperref}

\newcommand{\thybox}[1]{%
  {\setlength{\fboxsep}{0pt}\fbox{\includegraphics[width=\dimexpr\textwidth-2\fboxrule\relax,%
     height=.84\textheight,keepaspectratio]{#1}}}}
\newcommand{\thymbox}[1]{%
  {\setlength{\fboxsep}{0pt}\fbox{\includegraphics[width=\dimexpr\textwidth-2\fboxrule\relax,%
     height=.85\textheight,keepaspectratio]{#1}}}}
\newcommand{\thyfig}[2]{%
  \begin{figure}[tp]\centering\thymbox{#1}\caption{#2}\label{fig:#1.thy}\end{figure}}

\newcommand{\thyfigH}[2]{%
  \begin{figure}[H]\centering\thybox{#1}\caption{#2}\label{fig:#1.thy}\end{figure}}
\newcommand{\thyfigpartH}[3]{%
  \begin{figure}[H]\centering\thybox{#1_#2}\caption{#3}\label{fig:#1.thy.#2}\end{figure}}

\newif\ifextended
\extendedtrue
\ifextended
  \newcommand{\extonly}[1]{#1}%      shown only in the extended version
  \newcommand{\procorext}[2]{#1}%    extended: first argument
\else
  \newcommand{\extonly}[1]{}%        hidden in the proceedings version
  \newcommand{\procorext}[2]{#2}%    proceedings: second argument
\fi

\let\oldthebibliography\thebibliography
\renewcommand{\thebibliography}[1]{%
  \oldthebibliography{#1}%
  \setlength{\itemsep}{0pt}%
  \setlength{\parsep}{0pt}%
  \setlength{\topsep}{0pt}%
}

\begin{document}
\sloppy

%% Keep the ORCID numbers visible but render them a touch smaller and unbreakable,
%% so the two-author line fits on a single line (llncs' default full-size
%% superscripts overflow into a second line).
\renewcommand{\orcidID}[1]{\unskip$^{[\mbox{\tiny #1}]}$}

\title{Monadic Second-Order Logic in HOL: Deep and Shallow with Automated Faithfulness\extonly{\\(Extended Preprint)}}
\titlerunning{Monadic Second-Order Logic in HOL}

\author{Christoph Benzm\"uller\inst{1,2}\orcidID{0000-0002-3392-3093}\thanks{Dedicated to David Basin, whom I first met in Saarbr\"ucken in the mid-1990s, over a shared enthusiasm for automated reasoning, formal methods, and expressive logics.} \and
        Daniel Kirchner\inst{1}\orcidID{0000-0001-9229-1148}}
\authorrunning{C. Benzm\"uller \& D. Kirchner}
\institute{University of Bamberg, Bamberg, Germany\\
  \email{christoph.benzmueller@uni-bamberg.de}, \email{daniel.kirchner@uni-bamberg.de}
  \and Freie Universit\"at Berlin, Berlin, Germany}

\maketitle

\begin{abstract}
In Isabelle/HOL, we apply the deep-and-shallow embedding methodology of our
prior work to \emph{monadic second-order logic} (MSO). Three embeddings are
developed side by side: a \emph{deep} embedding (an inductive datatype with an
explicit satisfaction relation); a \emph{maximal-shallow} embedding that
translates the connectives and quantifiers directly into HOL, carrying the
interpretation and both assignments explicitly; and a
\emph{minimal-shallow} embedding --- a locale that fixes those parameters,
collapsing the formula type to $\mathit{bool}$. The enabling new ingredient is a
\emph{two-sorted} substitution apparatus --- capture-avoiding substitution,
renaming, and a substitution lemma per namespace --- in which each binder is
transparent for the other; faithfulness of all three embeddings is mechanised
and largely automated. Our central contribution is a fully mechanised two-sorted downward L\"owenheim--Skolem
theorem: the minimal embedding recovers deep validity relative to the
(countable) assignment ranges, and this range-relative reading is shown to
coincide with the general (Henkin-style) reading of MSO, whereas the standard
reading validates strictly more formulas, witnessed by comprehension. Both readings are
nonetheless recovered from the minimal embedding, differing only in the admitted
interpretations: all of them for the general reading, only the
\emph{elementary substructures} of the full model for the standard. We exercise the
embeddings on classical MSO landmarks: the Boolean-closure and graph schemata
hold under the full second-order domain yet fail in the minimal embedding,
making the dichotomy concrete, while reachability and $2$-colorability are
refuted throughout.
\end{abstract}

%% =================================================================
\section{Motivation and introduction}\label{sec:intro}
%% =================================================================

Monadic second-order logic (MSO) sits at the heart of a body of
results connecting logic, automata, and graph/tree structure
theory~\cite{Buchi1960,Thomas1997,Courcelle2012}. It extends first-order
relational logic with a single, but powerful, additional resource:
quantification over \emph{monadic} predicates, i.e.\ over sets of individuals.

On the algorithmic side, a substantial line of work led by David Basin
approaches MSO through automata and decision procedures: Basin and Klarlund use
the monadic second-order logic of finite strings, M2L(Str), and its canonical
string automata (the Mona system) to specify and verify hardware, obtaining a
decision procedure with counter-model generation~\cite{BasinKlarlund1995,BasinKlarlund1998}, and Ayari
and Basin develop bounded model construction for monadic second-order
logics~\cite{AyariBasin2000}. That line targets decidable fragments over strings
and trees; the present work is complementary --- it embeds full and general
(Henkin-style) MSO
into classical higher-order logic for interactive proof and metatheory, rather
than deciding a fixed fragment.

Embeddings in classical higher-order logic
(HOL)~\cite{church1940} are a versatile host for non-classical and expressive
logics: treating HOL as a universal meta-logic, a wide range of object logics
--- modal, deontic, epistemic, and others --- inherit its semantic
infrastructure and proof automation. This is the basis of \emph{LogiKEy}~\cite{ManyLogics2026,J48}, a logic-pluralistic
knowledge representation methodology in HOL, from computational metaphysics to
ethical and legal reasoning; it grew out of the early modal-logic
embeddings developed together with Paulson~\cite{J21,J23}.

\emph{Shallow} embeddings, which represent object-level formulae directly as
HOL terms (for MSO, by the standard translation of connectives and
quantifiers), give immediate access to Isabelle/HOL's proof automation --- the
prover \texttt{sledgehammer}, the finite-countermodel finder \texttt{nitpick}, and
\texttt{auto}. \emph{Deep} embeddings,
in which formulae are elements of an inductively defined datatype, are needed
for metalogical reasoning about syntax, substitution, or proof systems.

Shallow semantical embeddings of propositional and quantified modal logics in
HOL were developed by Benzm\"uller and Paulson~\cite{J21,J23}, initially without
a deep embedding or a mechanised faithfulness proof. We have recently
shown~\cite{C98} that \emph{both} embeddings can be carried out simultaneously,
in a single HOL theory, with mutual faithfulness proofs automated. That
development, however, only treats propositional modal logic, whose binders are
modal operators carrying no variable binding. The same methodology now spans a
range of logics; we develop it here for monadic second-order logic (MSO). We
take up MSO not for expressive strength but for the phenomena of its second-order
quantifier --- above all the standard-versus-general (Henkin) distinction behind
our central result --- and, on the technical side, for its two-sorted binding:
quantifying over individuals and sets of individuals at once, MSO's first- and
second-order existentials are the two binders our apparatus must handle.

\paragraph{Contributions.}
We contribute \textsf{MSOinHOL}, an Isabelle/HOL development of MSO with binary
relational atoms (no primitive equality), monadic predication atoms, the propositional connectives, a
first-order existential quantifier, and a monadic second-order existential quantifier.
The particular contributions include:
\begin{enumerate}
\item A \emph{deep} embedding (datatype $\mathcal{F}$) and two \emph{shallow}
ones: a maximal embedding carrying explicit dependencies on the interpretation,
the two domains, and the two assignments; and a minimal embedding that fixes the
interpretation and assignments as locale parameters, collapsing the formula type
to $\mathit{bool}$.
\item The minimal embedding is an Isabelle/HOL \emph{locale} \texttt{MinS};
it admits a faithfulness theorem \textsf{FaithfulMS\_all}
stating that quantifying over all locale interpretations recovers exactly deep
validity relative to the assignment ranges. Faithfulness of the maximal-shallow
embedding with respect to the deep one (\textsf{FaithfulSD}) is likewise
established and automated.
\item A \emph{two-sorted} substitution apparatus for the deep embedding ---
free/bound/fresh predicates, capture-avoiding substitution, alphabetic renaming,
the substitution lemma, and size-based induction --- developed once per variable
namespace, with each binder \emph{transparent} for the other.
\item Experiments confirming tautologies, the monadic \emph{comprehension}
schema (proved for an arbitrary formula, generalising two ad-hoc instances), and
\texttt{nitpick} countermodels for invalid schemata.
\item A suite of \emph{classical MSO landmarks} --- Boolean
closure~\cite{Buchi1960,Thomas1997}, monadic graph
operations~\cite{Courcelle2012}, reachability~\cite{BasinKlarlund1995}, and
$2$-colorability~\cite{Thomas1997} --- across all three embeddings: the closure
and graph schemata hold under the full second-order domain but \emph{fail} in the
general minimal embedding (complement and intersection with \emph{potential} \texttt{nitpick} countermodels,
the rest failing alike; recoverable via \textsf{Deep'\_to\_MinS}), while
reachability and $2$-colorability are refuted throughout.
\item The converse: range-relative validity \emph{coincides} with general
(Henkin-style) MSO validity, by a two-sorted L\"owenheim--Skolem construction, while the
standard reading validates strictly more formulas --- all formally proved.
\end{enumerate}

\paragraph{Organisation.}
\S\ref{sec:hol} fixes notation for HOL, \S\ref{sec:mso} recalls MSO, and
\S\ref{sec:embeddings} presents the three embeddings. \S\ref{sec:subst} outlines
the two-sorted substitution machinery, \S\ref{sec:faithful} the faithfulness
theorems, \S\ref{sec:surj} the surjectivity problem, and \S\ref{sec:exp} the
experiments; \S\ref{sec:concl} concludes and discusses related work. The complete
formal development is the Archive of Formal Proofs entry
\textsf{MSOinHOL}~\cite{MSOinHOL-AFP}.
\procorext{Rendered theory sources appear in App.~\ref{app:src} and the full
two-sorted L\"owenheim--Skolem proof in App.~\ref{app:proof}.}{An extended
version of this paper, containing the full two-sorted L\"owenheim--Skolem proof
and the rendered theory sources as appendices, is available on
arXiv~\cite{msoinhol-ext}.}

%% =================================================================
\section{Classical higher-order logic}\label{sec:hol}
%% =================================================================

We adopt the presentation of HOL given in~\cite{C98} without modification:
simply typed $\lambda$-calculus over base types~$o$ (truth values) and~$D$
(individuals), Henkin semantics, and the standard denotation
$\llbracket\cdot\rrbracket$ derived from a frame $\mathcal{D}$ and an
interpretation~$I$. HOL's proof calculi are sound and complete with respect to Henkin
models~\cite{henkin1950},\footnote{The landscape of
standard models, Henkin models, and Andrews' weaker $v$-complexes, and the role
of extensionality in fixing the class of Henkin models, is analysed by
Benzm\"uller, Brown, and Kohlhase~\cite{J6}; the shallow embeddings below inherit HOL's functional and
Boolean extensionality.} and Isabelle/HOL~\cite{Isabelle} mechanises
a fragment of HOL sufficient for all that follows. In
particular, the second-order domain of MSO is the HOL function type
$D\Rightarrow o$, so monadic sets (resp.\ predicates) are first-class HOL objects.

%% =================================================================
\section{Monadic second-order logic (MSO)}\label{sec:mso}
%% =================================================================
The syntax of MSO, given a (nonempty) signature $R$ of binary relation
symbols, a denumerable set $V$ of individual (first-order) variables, and a
denumerable set $V_2$ of monadic (second-order) variables, is fixed by the
grammar
\begin{equation}\label{eq:syntax}
\varphi,\psi := r(x,y) \mid X(x) \mid \neg\varphi \mid \varphi\wedge\psi
                \mid \exists x.\varphi \mid \exists X.\varphi
\end{equation}
for relation symbols $r\in R$, individual variables $x,y\in V$, and monadic predicate variables $X\in V_2$; the atom $X(x)$ asserts that the monadic predicate $X$ holds of the individual $x$ (extensionally, that $x$ lies in the set denoted by $X$). The further connectives $\vee$, $\supset$,
$\leftrightarrow$, the first-order universal $\forall x$, and the
second-order universal $\forall X$ are defined as usual by duality.

A \emph{model} is a tuple $\langle I,D,E\rangle$, where $D$ is a nonempty domain of individuals, $E$ is a nonempty collection of \emph{admissible} monadic predicates over $D$ (the second-order domain; extensionally a family $E\subseteq\mathcal{P}(D)$ of subsets), and the
interpretation $I$ assigns to each relation symbol $r$ a binary relation
$I(r)\subseteq D\times D$. A \emph{first-order assignment} $g\colon V\to D$
maps each individual variable to a domain element, and a \emph{second-order assignment} $G\colon V_2\to E$ maps each monadic predicate variable to an admissible predicate in~$E$; $g[x\leftarrow d]$ and $G\langle X\leftarrow S\rangle$ denote the obvious updates.

Truth in a model is given recursively by
\begin{align*}
\langle I,D,E\rangle,g,G &\models r(x,y)
        && \text{iff } I(r)(g(x),g(y))\\
\langle I,D,E\rangle,g,G &\models X(x)
        && \text{iff } G(X)\ \text{holds of}\ g(x)\\
\langle I,D,E\rangle,g,G &\models \neg\varphi
        && \text{iff not } \langle I,D,E\rangle,g,G \models \varphi\\
\langle I,D,E\rangle,g,G &\models \varphi\wedge\psi
        && \text{iff } \langle I,D,E\rangle,g,G \models \varphi
            \text{ and } \langle I,D,E\rangle,g,G \models \psi\\
\langle I,D,E\rangle,g,G &\models \exists x.\varphi
        && \text{iff } \exists d\in D.\, \langle I,D,E\rangle,g[x\leftarrow d],G \models \varphi\\
\langle I,D,E\rangle,g,G &\models \exists X.\varphi
        && \text{iff } \exists S\in E.\, \langle I,D,E\rangle,g,G[X\leftarrow S] \models \varphi
\end{align*}
A formula~$\varphi$ is \emph{valid} iff $\langle I,D,E\rangle,g,G\models\varphi$
for every model, every $g$ into~$D$, and every $G$ into~$E$. When $E=\mathcal{P}(D)$ (every monadic predicate is admissible) the second-order quantifier is the
\emph{standard} (full) one; restricting $E$ yields the various
\emph{general} (Henkin-style) readings. The principle bridging the two is
\emph{comprehension} --- that every property definable by a formula determines an
admissible set; it holds by fiat when $E=\mathcal{P}(D)$ but can fail once $E$ is
restricted. This is the \emph{monadic} fragment: only monadic predicate variables (not relation variables) may be quantified, and a unary property of individuals is rendered directly as a monadic predicate applied to its argument.

\paragraph{Scope.}
For minimality we restrict the first-order signature to
\emph{binary} relation symbols and omit primitive equality; unary properties are rendered directly as monadic predications $X(x)$. Thus the formalisation covers the equality-free binary-relational fragment of MSO, which suffices for all results below.

%% =================================================================
\section{Three embeddings of MSO in HOL}\label{sec:embeddings}
%% =================================================================

We now present the three embeddings side by side, after fixing the
preliminaries shared by all of them.

\subsection{Preliminaries}\label{sec:prelim}

Preliminaries shared by all embeddings (theory
\texttt{MSOinHOL\_preliminaries.thy}\extonly{, Fig.~\ref{fig:MSOinHOL_preliminaries.thy}})
declare the type $D$ (individuals) and the type synonyms
$R:=\mathit{nat}$ (relation symbols), $V:=\mathit{nat}$ (first-order
variables), $V_2:=\mathit{nat}$ (second-order variables), together with
$\mathcal{R}$, $\mathcal{I}$, the first-order assignment type
$\mathcal{E}:=V\Rightarrow D$, the second-order assignment type
$\mathcal{G}:=V_2\Rightarrow(D\Rightarrow\mathit{bool})$, and the two
domain-restriction types $\mathcal{D}:=D\Rightarrow\mathit{bool}$ and
$\mathcal{P}:=(D\Rightarrow\mathit{bool})\Rightarrow\mathit{bool}$. Two
variable-assignment update operators are introduced: $g[x\leftarrow d]$ for
the first-order assignment and $G\langle X\leftarrow S\rangle$ for the
second-order one, each with its companion lemmata
(\texttt{L1}--\texttt{L4} and \texttt{M1}--\texttt{M4}). The
set-as-predicate operations $\sqsubseteq$, $\sqcup$, \texttt{Range},
\texttt{Univ}, \texttt{into}, \texttt{onto} are polymorphic and hence shared
between both sorts.

\subsection{Deep embedding}\label{sec:deep}

The deep embedding (theory \texttt{MSOinHOL\_deep.thy},
Fig.~\ref{fig:MSOinHOL_deep.thy}) is a single inductive datatype
$\mathcal{F}$ whose constructors mirror the grammar~\eqref{eq:syntax}:
\texttt{AtmD} (relational atom $r^d(x,y)$), \texttt{PrdD} (monadic predication atom $X^{d}(x)$), \texttt{NegD}, \texttt{AndD}, the first-order
\texttt{ExD} ($\exists^d x.\varphi$), and the second-order \texttt{ExD2}
($\exists^d_2 X.\varphi$). The relative-truth predicate
$\langle I,D,E\rangle,g,G \models^d \varphi$ is declared in mixfix syntax and
defined by primitive recursion: the predication clause evaluates the monadic
predicate $G(X)$ at the individual $g(x)$; the first-order existential ranges over the explicit domain~$D$;
the second-order existential ranges over the explicit collection~$E$ of
admissible sets; the rest is standard. Validity \texttt{ValD} closes
universally over models, over domain-respecting first-order assignments
(\texttt{g into D}), and over domain-respecting second-order assignments
(\texttt{G into E}). An auxiliary full-domain notion $\models^{d'}$ fixes $D=\mathtt{Univ}$ and
$E=\mathtt{Univ}$; it is the \emph{standard} reading, the target of \S\ref{sec:surj},
and a convenience for \texttt{nitpick}; since $\models^d\varphi$ implies $\models^{d'}\varphi$,
any counterexample for $\models^{d'}$ refutes $\models^d$ as well.

\paragraph{Reading the clauses.}
The two existential clauses are the only places where the model is consulted
nontrivially: each quantifies a witness over its domain ($d$ over~$D$, $S$
over~$E$) and re-evaluates under the corresponding assignment update. All
remaining connectives ($\vee$, $\supset$, $\leftrightarrow$, and the two
universals) are definitional abbreviations, unfolded automatically through the
rewrite bag \texttt{DefD}; the same discipline (\texttt{DefS}, \texttt{DefM})
governs the two shallow embeddings, so the three notions of validity differ only
in their primitives.

\thyfig{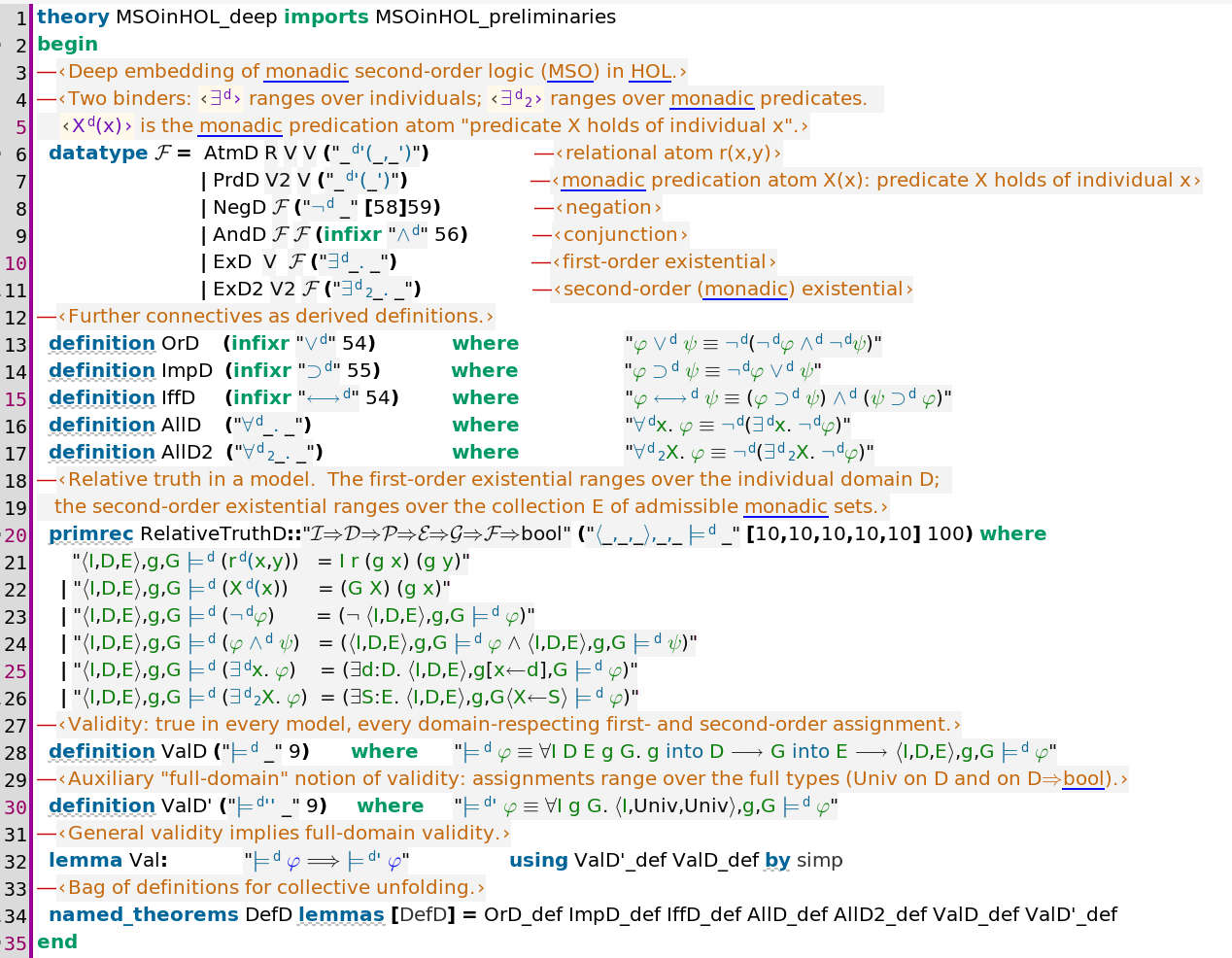}{Deep embedding of MSO in HOL: the datatype $\mathcal{F}$ with the predication atom and the two binders, the derived connectives, the relative-truth predicate \texttt{RelativeTruthD}, and the validity definitions \texttt{ValD}, \texttt{ValD'}.}

\subsection{Maximal shallow embedding}\label{sec:maxshallow}

The maximal shallow embedding (theory \texttt{MSOinHOL\_shallow.thy},
Fig.~\ref{fig:MSOinHOL_shallow.thy}) lifts every primitive of MSO to a
$\lambda$-term abstracted over the explicit dependencies $I$, $D$, $E$, $g$,
$G$. The type of formulae is
$\sigma := \mathcal{I}\Rightarrow\mathcal{D}\Rightarrow\mathcal{P}\Rightarrow
\mathcal{E}\Rightarrow\mathcal{G}\Rightarrow\mathit{bool}$. Atoms consult the
interpretation; the predication case applies the second-order assignment; the
first-order existential pairs a meta-level HOL existential with a first-order
assignment update; and the second-order existential pairs a meta-level HOL
existential with a second-order assignment update. The pattern is identical to
that of~\cite{C98}; the only difference is that the formula type $\sigma$
carries two assignments and two domain restrictions, because both the
first-order (individual) and the second-order (set) quantifier must be
accommodated. We write $\models^s$ for the resulting validity and $\models^{s'}$
for its full-domain variant ($D=\mathtt{Univ}$ and $E=\mathtt{Univ}$), the maximal-shallow analogues of
$\models^d$ and $\models^{d'}$.

\thyfig{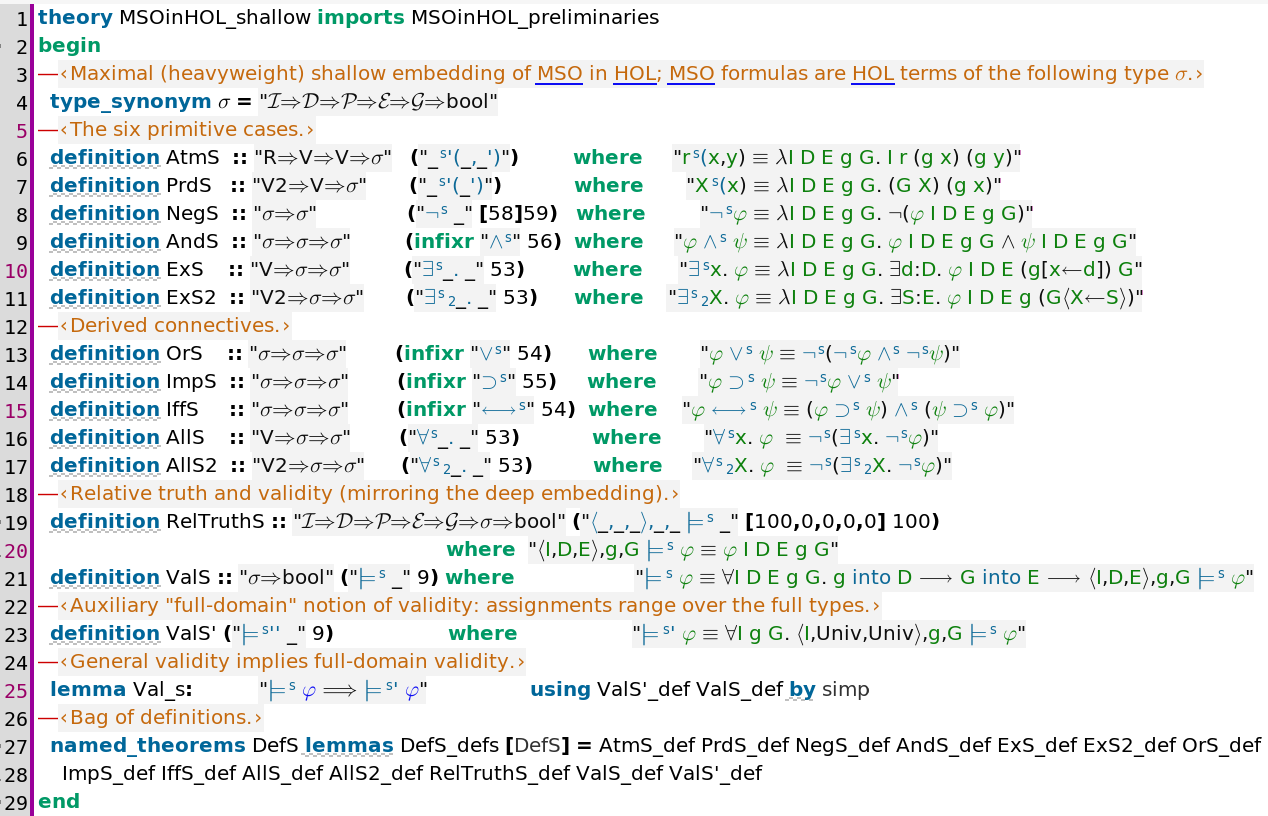}{Maximal shallow embedding
(\texttt{MSOinHOL\_shallow.thy}): the formula type $\sigma$, the six primitive
cases (the two existentials thread $g$ and $G$ respectively), the derived
connectives, and validity.}

\subsection{Minimal shallow embedding (locale-based)}\label{sec:minshallow}

The maximal-shallow type carries every dependency as an explicit argument,
threaded throughout because each quantifier \emph{updates} an assignment in
its recursive call ($g[x\leftarrow d]$ for $\exists x$, $G\langle X\leftarrow
S\rangle$ for $\exists X$). The minimal embedding avoids this: each MSO
existential is encoded by HOL's own binder over the matching variable
namespace ($V$, $V_2$), and atoms read their variables' denotations from the
assignment ($X^{m}(x)$ is $(GG\,X)(gg\,x)$), so a bound variable is a symbol
looked up in the assignment rather than a witness installed into an updated
one. The recursion therefore never changes $II$, $gg$, or $GG$; this, with
the absence of any world dependency, lets us promote them to locale
parameters, and the formula type reduces to $\mathit{bool}$.

We package the minimal embedding as an Isabelle/HOL \emph{locale}
\texttt{MinS}~\cite{ballarin2014} (theory
\texttt{MSOinHOL\_\allowbreak shallow\_\allowbreak minimal\_\allowbreak locale.thy},
Fig.~\ref{fig:MSOinHOL_shallow_minimal_locale.thy}), fixing $(II,gg,GG)$.
The six primitive cases are \texttt{AtmM}, \texttt{PrdM} (predication $X^{m}(x)$ defined as $(GG\,X)(gg\,x)$), \texttt{NegM},
\texttt{AndM}, the first-order binder \texttt{ExM} over~$V$, and the
second-order binder \texttt{ExM2} over~$V_2$. As the formula type is
$\mathit{bool}$, relative truth and validity \texttt{ValM} (written $\models^m$) are the identity.
Both binders are declared as HOL
\emph{binders}, so that \texttt{sledgehammer} treats second-order MSO goals
as ordinary HOL goals: bound first- and second-order variables become
HOL-level bound variables of types $V$ and $V_2$, and the external provers
reason about them directly without unfolding the deep-embedding syntax.

For practical experimentation the locale is also instantiated once at the
global level (theory \texttt{MSOinHOL\_shallow\_minimal.thy}%
\extonly{, Fig.~\ref{fig:MSOinHOL_shallow_minimal.thy}}): three
uninterpreted constants \texttt{II}, \texttt{gg}, \texttt{GG} are introduced
and a \texttt{global\_interpretation} re-issues the locale notation at the
global level. The locale formulation remains the single source of all
theorems; the global interpretation propagates them to the constants without
proof duplication.

\thyfig{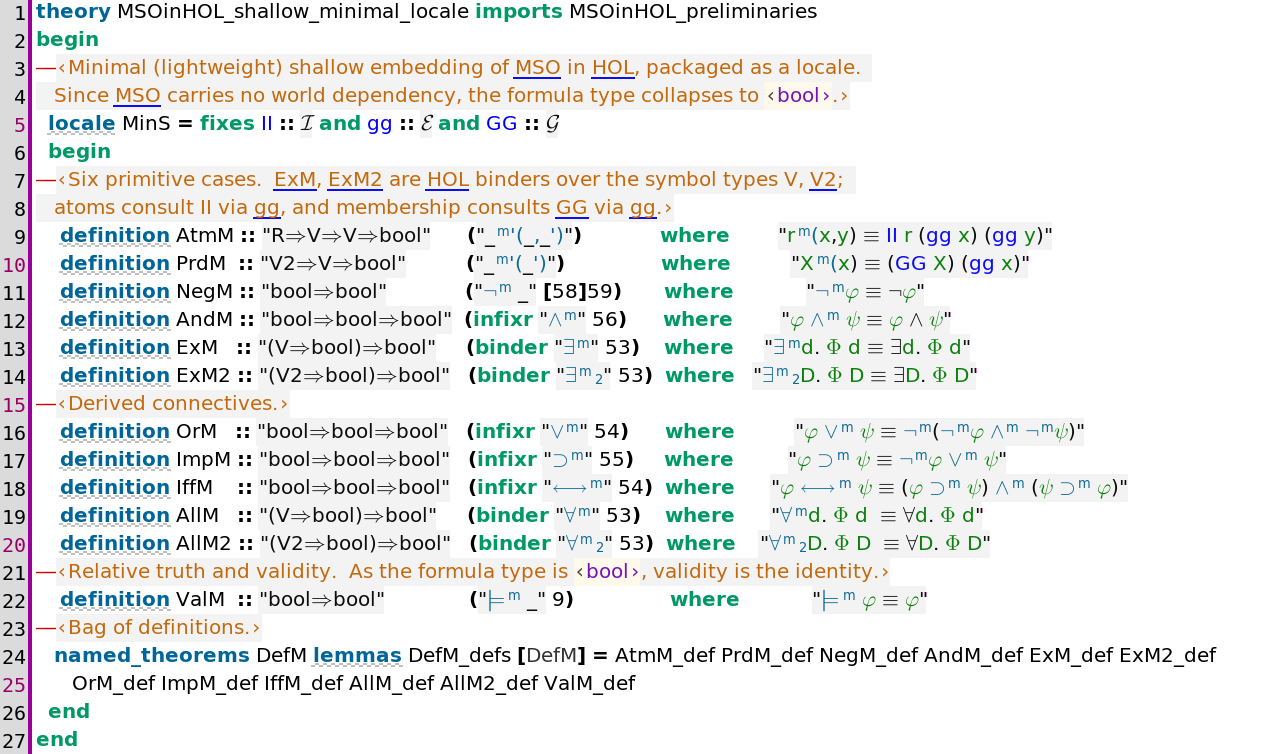}{The locale \texttt{MinS}
(\texttt{MSOinHOL\_shallow\_minimal\_locale.thy}) packaging the minimal shallow
embedding: parameters $II$, $gg$, $GG$; the six primitive cases including the
HOL binders \texttt{ExM} (over $V$) and \texttt{ExM2} (over $V_2$); the derived
connectives; and validity.}

\paragraph{Maximal vs.\ minimal: the trade-off.}
The two shallow embeddings encode the same object logic but package the semantic
parameters oppositely. The \emph{maximal} embedding keeps the interpretation and
both assignments explicit in each translated formula (type~$\sigma$): the
parameters are under direct control and the translation lines up term-for-term
with the deep embedding, which makes it the natural \emph{faithfulness bridge}
(\textsf{FaithfulSD}) and the place to express the general reading $\models^{s}$
and its full-domain standard variant $\models^{s'}$; the price is weight, as every subterm threads
$I,D,E,g,G$. The \emph{minimal} embedding hides those parameters in the locale
\texttt{MinS}, so a translated formula becomes an \emph{ordinary HOL proposition}
of type~$\mathit{bool}$ with each MSO binder a native HOL binder:
\texttt{sledgehammer} and \texttt{nitpick} act on it directly, which makes it the
convenient vehicle for downstream reasoning and for comparing classes of
interpretations --- general vs.\ standard (\S\ref{sec:surj}) --- at the cost of
leaving the semantic parameters implicit.

%% =================================================================
\section{Two-sorted substitution machinery}\label{sec:subst}
%% =================================================================

Substitution lemmas are needed for the deep$\leftrightarrow$minimal-shallow
faithfulness proof, to handle the two existential-quantifier cases. Since MSO
has two variable namespaces, we develop \emph{two} parallel substitution
apparatuses, \emph{Part~A} (first-order) and \emph{Part~B} (second-order), in
theory \texttt{MSOinHOL\_deep\_subst\_lemma.thy}%
\extonly{ (Figs.~\ref{fig:MSOinHOL_deep_subst_lemma.thy.1}--\ref{fig:MSOinHOL_deep_subst_lemma.thy.3})}.
We use named variables and explicit capture-avoiding substitution rather than a
nominal or de~Bruijn infrastructure, keeping the two namespaces and their mutual
transparency visible in the formal text.

\paragraph{The two apparatuses.}
Each part provides the standard components --- free/bound/fresh predicates, a
variable-for-variable substitution \texttt{Subst} (written
$[x{\leftarrow}z](\varphi)$; second-order \texttt{Subst2}), a substitutability
predicate, the \textsf{SubstitutionLemma} (resp.\ \textsf{SubstitutionLemma2},
whose semantic side is the update $G\langle X\leftarrow G\,Z\rangle$), and an
$\alpha$-renaming \texttt{ren\_for\_subst} composing into the safe substitution
\texttt{ren\_subst}, written $[x\!\leftarrow_r\!z](\varphi)$. The size-based
induction principles \texttt{SInduct} and \texttt{QInduct} recurse on \emph{both}
binders. The outcome used downstream is \texttt{L27}/\texttt{N27} (safe
substitution commutes with assignment update); \texttt{L29}/\texttt{N29} record
the resulting equation for surjective assignments --- a deep existential is a HOL
existential over the variable symbols --- exactly the shape the minimal
translation realises.

\paragraph{Transparency across namespaces.}
A binder of one sort contributes \emph{nothing} to the machinery of the other:
binding no variable of that namespace, it descends through every predicate,
substitution, renaming, and induction by a single uniform clause, structurally
identical to the negation clause and dispatched by the same
\texttt{auto}/\texttt{simp}. This holds \emph{twice over} --- the second-order
binder \texttt{ExD2} is transparent for Part~A and the first-order \texttt{ExD}
for Part~B --- the only genuinely new leaf being the predication atom $X(x)$,
which contributes a free first-order occurrence ($x$) to Part~A and a free
second-order one ($X$) to Part~B. The $\alpha$-renaming is what keeps
substitution sound across a binder: replacing $x$ by $z$ in
$\exists^d z.\,r^d(x,z)$ would \emph{capture} the free $z$, which
\texttt{ren\_subst} prevents by first renaming the bound $z$ to a fresh symbol.
We additionally isolate the single-binder semantic core as
\textsf{rename\_eval}/\textsf{rename\_eval2} (theory
\texttt{MSOinHOL\_subst\_extras.thy}%
\extonly{, Fig.~\ref{fig:MSOinHOL_subst_extras.thy}}).

%% =================================================================
\section{Faithfulness, automated}\label{sec:faithful}
%% =================================================================

An embedding is \emph{faithful} when validity transfers in both directions ---
a formula is valid in MSO exactly when its HOL translation is --- so the
embedding neither loses nor invents theorems; we establish this, largely automatically, for all three.

Two mappings translate deep formulae into shallow ones (theory
\texttt{MSOinHOL\_\allowbreak faithfulness\_\allowbreak locale.thy},
Fig.~\ref{fig:MSOinHOL_faithfulness_locale.thy}): \texttt{DpToShS} (deep to
maximal-shallow), written
$\llbracket\cdot\rrbracket$, and \texttt{DpToShM} (deep to minimal-shallow),
written
$(\!\lvert\cdot\rvert\!)$. The first recurses transparently through every
primitive, since both source and target use named-variable binders for both
sorts. The second, declared inside the locale \texttt{MinS}, bridges the two
named-variable binders of the deep embedding to the two HOL binders of the
minimal embedding: its first-order existential clause inserts a first-order
safe substitution $[v\!\leftarrow_r\!d](\varphi)$, and its second-order
existential clause inserts a second-order safe substitution
$[V\!\leftarrow_{r2}\!D](\varphi)$.

\paragraph{Pointwise and validity-level theorems.}
The deep$\leftrightarrow$maximal-shallow correspondence \textsf{FaithfulSDlem}
is derived by a single induction on~$\varphi$ with
\texttt{auto simp:\,DefS\,DefD}; its validity-level form \textsf{FaithfulSD} ---
the statement for closed validity rather than for a fixed model and
assignment --- then follows by unfolding the two validity definitions
\texttt{ValD\_def} and \texttt{ValS\_def}. The
deep$\leftrightarrow$minimal-shallow correspondence is first shown pointwise,
\textsf{FaithfulMDlem}, by a single \texttt{induct} with rule
\texttt{QInduct}: \texttt{simp add:\,DefD\,DefM} discharges the propositional
and atomic cases, and the two existential cases are closed by \texttt{blast}
after the safe-substitution lemmata \texttt{L27} and \texttt{N27} rewrite the
translated binders. As the minimal validity \texttt{ValM} is the identity,
\textsf{FaithfulMD} coincides with \textsf{FaithfulMDlem}, and
\textsf{FaithfulMS} follows by composition with \textsf{FaithfulSDlem}.

\paragraph{Faithfulness across all locale interpretations.}
The locale form admits the all-interpretations theorem \textsf{FaithfulMS\_all}: a
formula holds in the minimal embedding under \emph{every} interpretation
$(II,gg,GG)$ of \texttt{MinS} iff it holds, in the deep embedding, in every
model whose first- and second-order domains are exactly the ranges of the
assignments,
\[
\bigl(\forall II\,gg\,GG.\ \models^m (\!\lvert\varphi\rvert\!)\bigr)
\;=\;
\bigl(\forall I\,g\,G.\ \langle I,\mathrm{Range}\,g,\mathrm{Range}\,G\rangle,g,G \models^d \varphi\bigr).
\]
The proof is a single \texttt{simp} with \textsf{FaithfulMD}. The forward
direction toward full deep validity, \textsf{Deep\_to\_MinS}, holds outright:
a deep-valid formula transfers to every minimal interpretation. The converse
is exactly the surjectivity problem --- whether the countable assignment ranges
can be made to exhaust the full second-order domain --- to which we now turn.

\thyfig{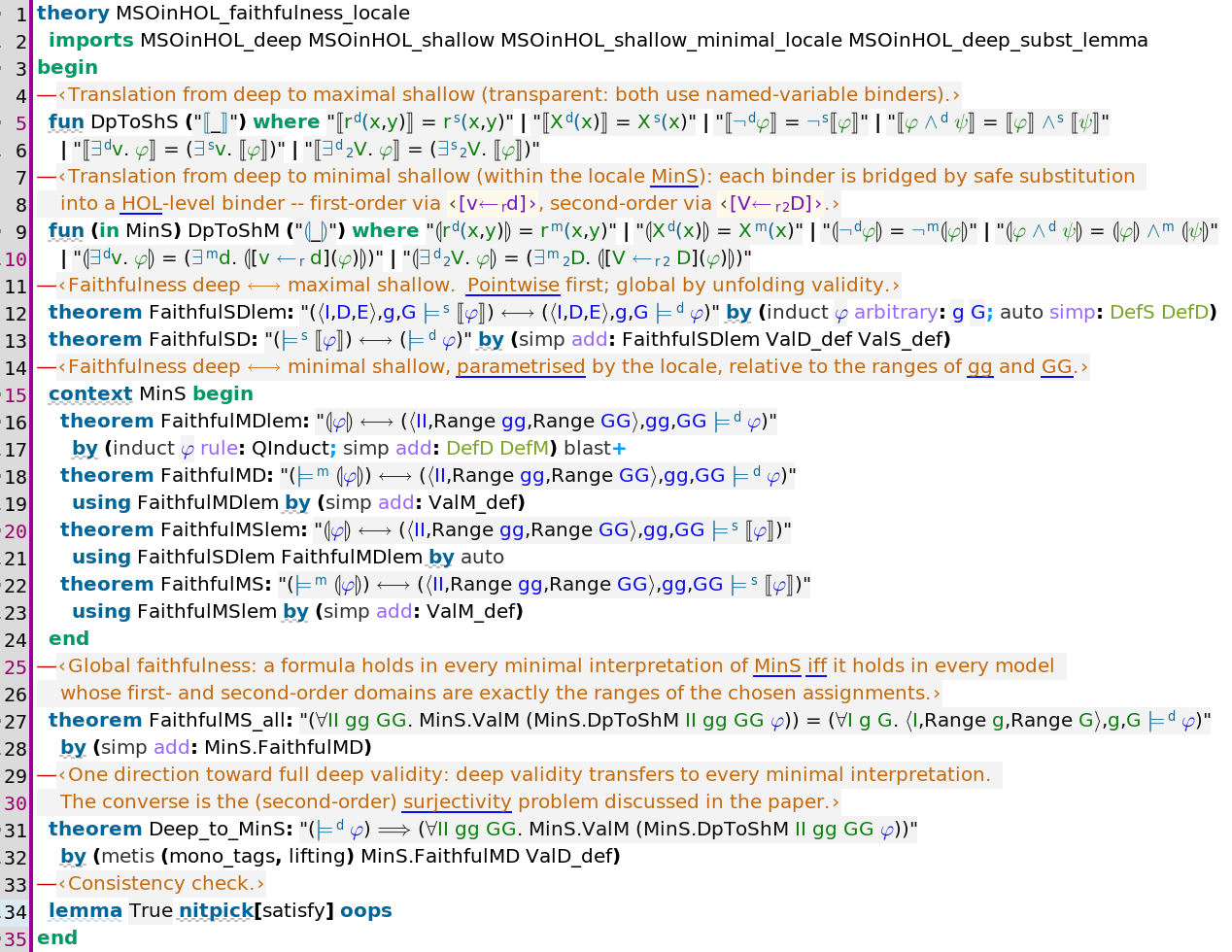}{Faithfulness: the translations $\llbracket\cdot\rrbracket$ and $(\!\lvert\cdot\rvert\!)$ (the latter bridging \emph{both} binders via the two safe substitutions), the deep$\leftrightarrow$maximal-shallow theorems, the locale-based deep$\leftrightarrow$minimal-shallow correspondences, the all-interpretations theorem \textsf{FaithfulMS\_all}, and the one-directional bridge \textsf{Deep\_to\_MinS}.}

\paragraph{What \textsf{FaithfulMS\_all} buys.}
\textsf{FaithfulMS\_all} makes the minimal embedding a \emph{proof}
device, not merely a definitional convenience. To establish a deep validity one
discharges the corresponding goal in the minimal embedding --- a plain HOL
formula with two ordinary binders, on which \texttt{sledgehammer} and
\texttt{auto} operate directly --- and the theorem transports the result to the
deep syntax over the range-determined model class; conversely, any genuine countermodel to the minimal goal refutes the deep
formula over that class (\texttt{nitpick}'s finite models are here only
\emph{potential}, \S\ref{sec:exp-classic}). The only gap is the passage from the range-determined model class to
the full general (Henkin-style) class, which \S\ref{sec:surj} closes; for the
standard reading it shows that unrestricted all-interpretations faithfulness is
impossible, while an elementary-substructure restriction recovers it exactly.

%% =================================================================
\section{The surjectivity problem, second-order}\label{sec:surj}
%% =================================================================

In the minimal-shallow embedding the first-order binder ranges over the
namespace $V=\mathit{nat}$ through $gg\colon V\to D$, with effective range
$\mathrm{Range}\,gg$; the second-order binder ranges over $V_2=\mathit{nat}$
through $GG\colon V_2\to(D\Rightarrow\mathit{bool})$, with effective range $\mathrm{Range}\,GG$.
Both ranges are countable, whereas the deep embedding lets the existentials
range over the full domains $D$ and $E$, so faithfulness against the full
domains does not follow from \textsf{FaithfulMD}. By \textsf{FaithfulMS\_all} the
missing converse is precisely
\[
  \bigl(\forall I\,g\,G.\ \langle I,\mathrm{Range}\,g,\mathrm{Range}\,G\rangle,g,G\models^d\varphi\bigr)
  \ \Longrightarrow\ \models^d\varphi .
\]

The decisive observation is that the general validity $\models^d$ imposes
\emph{no} closure on the admissible collection $E$: it quantifies over arbitrary
families with $G\ \mathtt{into}\ E$. That is exactly the \emph{general}
(Henkin-style) reading of MSO,\footnote{\label{fn:hierarchy}The reading $\models^d$ is the class of
structures $\langle I,D,E\rangle$ with an \emph{arbitrary} admissible family
$E\subseteq\mathcal P(D)$ and no closure condition; terminology for it is not
uniform. Shapiro~\cite{Shapiro1991} calls these \emph{Henkin models}, reserving
\emph{faithful Henkin model} for one that additionally satisfies comprehension and
choice, and \emph{full model} for the standard case $E=\mathcal P(D)$.
V\"a\"an\"anen~\cite{VaananenSEP} instead calls the comprehension-closed
structures \emph{Henkin models} and the unrestricted ones \emph{general models};
Manzano~\cite{Manzano1996} calls the unrestricted ones \emph{frames} and the
comprehension-closed ones \emph{general structures}, after Henkin's~\cite{henkin1950}
original \emph{general models}, which are closed under definability. We write
``general (Henkin-style)'' for $\models^d$ to stay neutral: it is Shapiro's
\emph{Henkin} class, V\"a\"an\"anen's \emph{general} class, and Manzano's
\emph{frames}. On every convention
the levels nest as full/standard $\subseteq$ comprehension-closed
(\emph{faithful Henkin}/\emph{Henkin}) $\subseteq$ general, with $\models^d$ the
widest; the diagonal comprehension instance separates $\models^d$ from the
comprehension-closed level, while the categoricity phenomenon separates that level
from the full one. Via the reduction of this semantics to many-sorted first-order
logic (Shapiro~\cite[Ch.~4]{Shapiro1991}; Manzano~\cite{Manzano1996}),
$\models^d$ inherits the downward L\"owenheim--Skolem theorem directly.} which is nothing but
two-sorted first-order logic over the
structure with individuals as one sort, the admissible sets as a second sort,
the relation symbols as binary relations on the first sort, and a membership
relation $d\in S \Leftrightarrow S(d)$ between the sorts. First-order logic has
the downward L\"owenheim--Skolem property; the apparent ``second-order''
difficulty is an artefact of conflating $\models^d$ with the \emph{standard}
reading $\models^{d'}$ treated below.

\paragraph{The converse (range-relative $\Rightarrow$ deep validity) holds, for the general reading.}
Together with \textsf{Deep\_to\_MinS}, this gives the following exact
characterisation of range-relative validity.
\begin{theorem}\label{thm:rangehenkin}
For every formula $\varphi$, range-relative validity coincides with general
\textup{(}Henkin-style\textup{)} deep validity:
\[
  \bigl(\forall I\,g\,G.\ \langle I,\mathrm{Range}\,g,\mathrm{Range}\,G\rangle,g,G\models^d\varphi\bigr)
  \quad\Longleftrightarrow\quad \models^d\varphi .
\]
Equivalently, via \textsf{FaithfulMS\_all}, the minimal-shallow embedding is
faithful \emph{exactly} to general \textup{(}Henkin-style\textup{)} MSO validity.
\end{theorem}
\noindent The proof (\procorext{App.~\ref{app:proof}}{extended
version~\cite{msoinhol-ext}}) is a two-sorted downward
L\"owenheim--Skolem construction. From a countermodel one extracts a
\emph{countable} elementary sub-model --- one satisfying exactly the same
formulae --- containing the assignment ranges by the Tarski--Vaught test (a
witness-reflection criterion that guarantees elementarity), drawing
set-witnesses honestly from $E$ --- so that no
comprehension is assumed and the sub-model stays inside the general class --- and
then reindexes the countably many variables to surject onto the countable
sub-domains while fixing the finitely many free variables of $\varphi$, using a
coincidence lemma. The second-order domain $E$ is never replaced by a powerset;
comprehension-closure is neither assumed nor produced: the hull inherits
exactly the closure properties, if any, of the ambient~$E$.

\paragraph{Unrestricted minimal interpretations do not capture the standard reading.}
No L\"owenheim--Skolem argument can strengthen this to the \emph{standard}
validity $\models^{d'}$ (where $E=\mathcal P(D)$), because the two notions
genuinely differ. Writing $\models^d_r$ for range-relative validity --- the
left-hand side of Theorem~\ref{thm:rangehenkin}, there identified with
$\models^d$ --- we have
\[
  \models^d_r \;=\; \models^d \;\subsetneq\; \models^{d'} .
\]
The inclusion $\models^d\subseteq\models^{d'}$ is immediate --- the standard
model is itself one of the general models quantified over by $\models^d$, which
is the content of the lemma we name \textsf{Val} --- while its strictness is
witnessed by comprehension: the
instance $\exists X.\,\forall x.\,(X(x)\leftrightarrow P(x,x))$ is
$\models^{d'}$-valid (\textsf{comprehension\_atom}) but not $\models^d$-valid:
\textsf{Standard\_strictly\_stronger} exhibits an explicit general model whose
restricted $E$ omits the diagonal set $\{d : P(d,d)\}$. In the three-level hierarchy
of Footnote~\ref{fn:hierarchy}, this instance separates the \emph{general} reading
$\models^d$ from the comprehension-closed \emph{Henkin} reading, and hence
\emph{a fortiori} from the full reading $\models^{d'}$; the residual gap, between
the Henkin and the full reading, reflects the familiar \emph{categoricity}
phenomenon for second-order logic --- the full reading pinning models down up to
isomorphism. We do not formalise this second gap; our machine-checked results
concern the embeddings, their faithfulness, and the L\"owenheim--Skolem connection
between the range-relative and general readings, while the separation
$\models^d\subsetneq\models^{d'}$ from the standard reading is witnessed here by
the simpler diagonal comprehension instance above.

The analysis thus resolves into a dichotomy: the
converse is a \emph{theorem} for the general (Henkin-style) reading and a provable
\emph{non-reduction} for the standard reading. The supplementary theories
\texttt{MSOinHOL\_\allowbreak lowenheim\_\allowbreak skolem\_\allowbreak lemmas.thy}
and \texttt{MSOinHOL\_\allowbreak lowenheim\_\allowbreak skolem.thy} carry this out
--- the coincidence lemma \textsf{coincidence}, the surjective reindexing
\textsf{reindex}, the assembled converse \textsf{RangeValid\_imp\_ValD} (with
corollaries \textsf{RangeValid\_iff\_ValD} and \textsf{Faithful\_to\_Henkin}), the
negative \textsf{Standard\_strictly\_stronger} with its explicit diagonal
countermodel, and the two-sorted elementary-substructure hull \textsf{ls\_hull},
a countable Tarski--Vaught hull \textsf{skolem\_hull} whose truth-preservation
lemma \textsf{truth\_pres} drives the two-sorted construction --- and, built on the
same hull, the elementary-substructure relation \textsf{ElementarySubstructure}, the
first-class downward L\"owenheim--Skolem theorem \textsf{DownwardLowenheimSkolem}, and
the standard-reading characterisation \textsf{Deep'\_to\_MinS}. All are complete
structured proofs, machine-checked under Isabelle2025-2; the
construction is given in \procorext{App.~\ref{app:proof}}{the extended
version~\cite{msoinhol-ext}}.

\paragraph{The standard reading, captured exactly by elementary interpretations.}
The non-reduction says only that no \emph{range-relative} --- equivalently,
all-interpretations --- faithfulness can reach $\models^{d'}$. It does not say the
standard reading escapes the minimal embedding: it does not. Writing
$\langle I',D',E'\rangle\subseteq_{\!E}\langle I,D,E\rangle$ for the two-sorted
\emph{elementary substructure} relation \textsf{ElementarySubstructure} ($D'\sqsubseteq D$,
$E'\sqsubseteq E$, and the two models satisfy the same formulae under every assignment
landing in the smaller pair), the development now isolates the Tarski--Vaught hull as a
\emph{first-class} downward L\"owenheim--Skolem theorem \textsf{DownwardLowenheimSkolem}:
every countable nonempty seed sub-pair of a model extends to a countable elementary
substructure. Using it, the strongest faithfulness theorem \textsf{Deep'\_to\_MinS}
characterises the standard reading \emph{exactly},
\begin{multline*}
  \models^{d'}\varphi
  \ \Longleftrightarrow\
  \forall II\,gg\,GG.\\
  \langle II,\mathrm{Range}\,gg,\mathrm{Range}\,GG\rangle
  \subseteq_{\!E}\langle II,\mathrm{Univ},\mathrm{Univ}\rangle
  \ \longrightarrow\ \models^m(\!\lvert\varphi\rvert\!),
\end{multline*}
where $\langle II,\mathrm{Univ},\mathrm{Univ}\rangle$ is the full model ($E=\mathcal P(D)$).
Both readings of MSO are therefore recovered from the minimal embedding, differing only in the admissible
interpretations: $\models^d$ quantifies over \emph{all} minimal interpretations
(Theorem~\ref{thm:rangehenkin}, via \textsf{FaithfulMS\_all}), whereas $\models^{d'}$ quantifies
over the \emph{elementary} ones --- those whose countable range model is an elementary substructure
of the full model. The strict inclusion $\models^d\subsetneq\models^{d'}$ is exactly the statement
that this elementary restriction cannot be dropped, the diagonal comprehension instance witnessing
it. Concretely, any interpretation refuting a closure schema in the minimal embedding
--- the (potential) \texttt{nitpick} models of \S\ref{sec:exp} included --- is
necessarily \emph{non-elementary}: its countable range omits a set the schema needs; on the elementary interpretations the same schemas
become provable, in agreement with their $\models^{d'}$-validity. (They are thus not genuine
counterexamples to the standard reading, but artefacts of admitting non-elementary ranges; under the elementary-substructure restriction, the locale \textsf{MinS\_ES\_Univ}, each
would be a theorem.)

\paragraph{A recursion-theoretic strengthening.}
The non-reduction just established is model-theoretic; classically, the full MSO
models --- those with $E=\mathcal{P}(D)$ --- do not even form an elementary
class, as downward L\"owenheim--Skolem fails for them (not formalised here). A complementary,
recursion-theoretic obstruction is that standard validity $\models^{d'}$ is not
recursively enumerable --- a classical consequence of the expressive strength of
standard second-order logic (cf.\ Shapiro~\cite{Shapiro1991},
V\"a\"an\"anen~\cite{Vaananen2011}), which we do not formalise here --- so it is
the consequence set of no recursively axiomatisable first-order theory. This non-enumerability needs both ingredients, the standard reading \emph{and}
quantification over \emph{all} interpretations of the relation symbols: the general
reading is enumerable by two-sorted first-order completeness, and restricting the standard reading to suitable structures such as $(\mathbb{N},<)$,
the infinite binary tree, or finite words and trees makes it decidable, hence
enumerable~\cite{Buchi1960,Thomas1997,Courcelle2012}.

%% =================================================================
\section{Experiments}\label{sec:exp}
%% =================================================================

We illustrate the development with two main test files. The core file
\texttt{MSOinHOL\_\allowbreak experiments.thy} (Fig.~\ref{fig:MSOinHOL_experiments.thy})
collects tautologies, comprehension, and \texttt{nitpick} disproofs; the larger
file \texttt{MSOinHOL\_\allowbreak experiments\_\allowbreak classic.thy} replicates a suite of classical
MSO landmarks across all three embeddings (\S\ref{sec:exp-classic}).
Propositional tautologies (proved in all three embeddings with the d/s/m
discipline), first-order tautologies such as
$(\forall x.\,P(x,x))\supset(\exists x.\,P(x,x))$, and the predication tautology $\forall x.\,(X(x)\supset X(x))$ are dispatched by \texttt{simp} or \texttt{auto}
after unfolding \texttt{DefD}, \texttt{DefS}, or \texttt{DefM}.

\paragraph{Monadic comprehension.}
The principal second-order validity we verify is monadic \emph{comprehension},
stated over
the full second-order domain ($\models^{d'}$, $E=\mathtt{Univ}$). In its general
form it is a schema, parametric in an arbitrary formula~$\varphi$ in which the
comprehension variable~$X$ does not occur free:
\[
  \models^{d'}\ \exists X.\,\forall x.\,(X(x) \leftrightarrow \varphi).
\]
The witnessing set is exactly the property defined by~$\varphi$, namely the HOL
predicate
$\lambda d.\,\langle I,\mathtt{Univ},\mathtt{Univ}\rangle,g[x{\leftarrow}d],G\models^{d}\varphi$;
it always lies in $E$ because $E=\mathcal{P}(D)$ in the standard semantics, and
the side condition on~$X$ lets the second-order irrelevance lemma discharge the
$X$-slot inside~$\varphi$. We record this as the theorem
\textsf{comprehension\_schema} (supplementary theory
\texttt{MSOinHOL\_comprehension.thy}%
\extonly{, Fig.~\ref{fig:MSOinHOL_comprehension.thy}}); the two instances proved directly in
\texttt{MSOinHOL\_\allowbreak experiments.thy} --- the set of $P$-self-related individuals,
$\exists X.\,\forall x.\,(X(x)\leftrightarrow P(x,x))$, with witness
$\lambda d.\,P\,d\,d$, and the universally-true predicate
$\exists X.\,\forall x.\,X(x)$, with witness $\lambda d.\,\mathit{True}$ --- are
its corollaries. No second-order search is required, because the comprehension
term is itself an ordinary HOL predicate.

\paragraph{Disproofs by \texttt{nitpick}.}
\texttt{nitpick} finds finite countermodels for the invalid schemata considered
here; these are exploratory searches, terminated with \texttt{oops} rather than stored as theorems.
``Every monadic predicate holds somewhere'', $\forall X.\,\exists x.\,X(x)$, fails by the empty predicate; the comprehension instance asking for a predicate that both holds and fails of every individual, $\exists X.\,\forall x.\,(X(x)\wedge\neg X(x))$, is unsatisfiable hence invalid; and symmetry of $P$,
$\forall x\,\forall y.\,(P(x,y)\supset P(y,x))$, fails by a two-element model.
Each disproof uses the full-domain relation $\models^{d'}$.

\thyfig{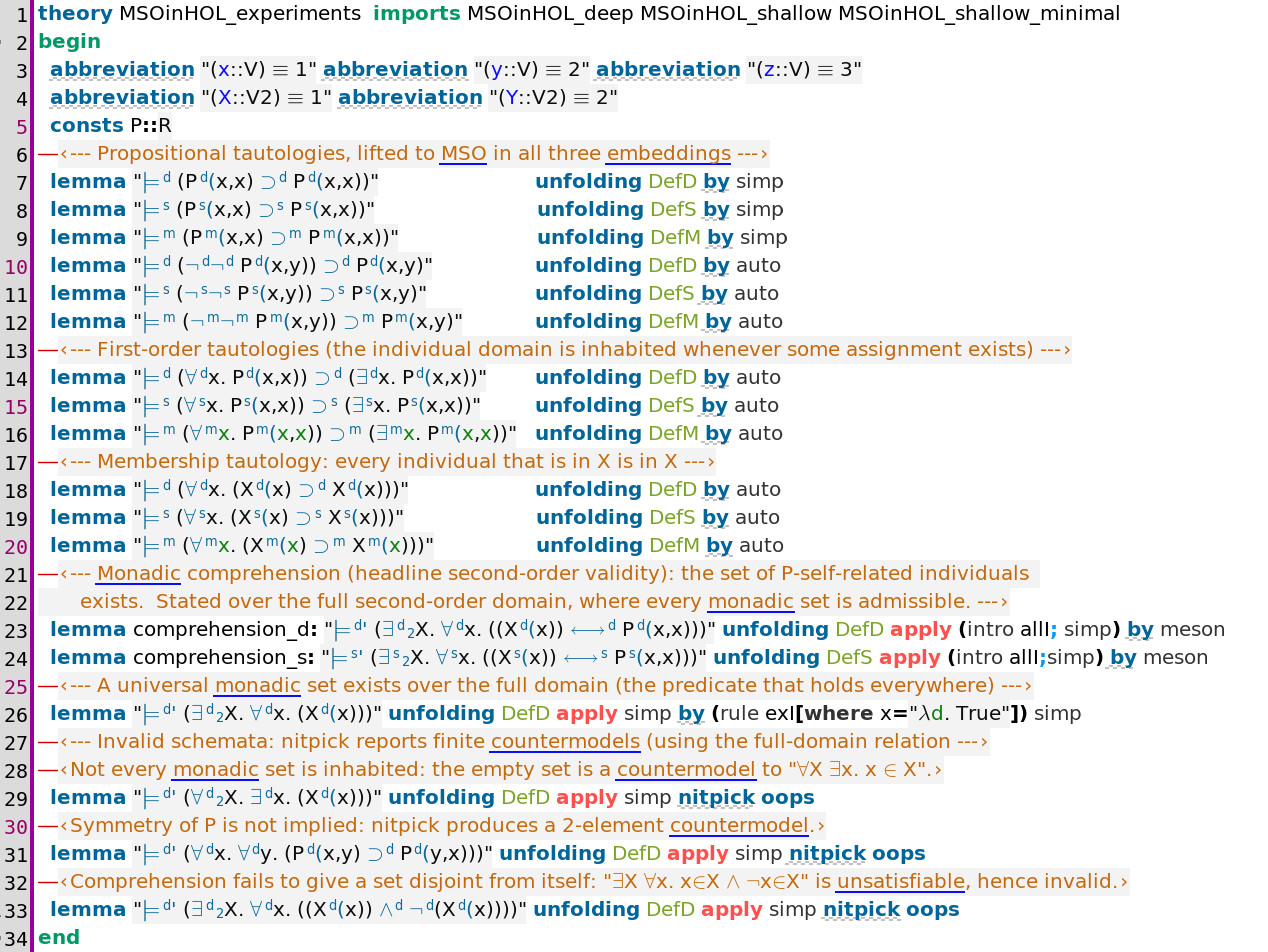}{Experiments: propositional and first-order tautologies in all three embeddings, the predication tautology, monadic comprehension and the universal-predicate validity over the full second-order domain, and \texttt{nitpick} disproofs of ``every predicate holds somewhere'', a self-contradictory predicate, and symmetry of $P$.}

\subsection{Classical MSO landmarks across the three embeddings}\label{sec:exp-classic}

A second, larger experiment file
(\texttt{MSOinHOL\_\allowbreak experiments\_\allowbreak classic.thy}%
\extonly{, Fig.~\ref{fig:MSOinHOL_experiments_classic.thy}})
revisits a sequence of textbook MSO
constructions and runs each one through all three embeddings. The point is not
the individual validities --- they are folklore --- but the \emph{pattern} of
their truth values across the five readings $\models^{d}$, $\models^{d'}$,
$\models^{s}$, $\models^{s'}$, $\models^{m}$, the operational face of the dichotomy
of \S\ref{sec:surj}; Table~\ref{tab:classic} collects the outcomes.

\begin{table}[t]
\centering
\caption{Classical MSO landmarks across the five readings. In the full-domain
columns $\models^{d'},\models^{s'}$ ($E=\mathtt{Univ}$) and the minimal column
$\models^{m}$, $\checkmark$ is a validity proof (an explicit second-order witness
for $\models^{d'}/\models^{s'}$), $\times$ a \texttt{nitpick}-found countermodel
(in the $\models^{m}$ column a \emph{potential} one, see below), and
``--'' a case not separately tested (failing as complement/intersection do). The
general readings $\models^{d},\models^{s}$ are not tested independently: they
coincide with $\models^{m}$ by \textsf{FaithfulSD} ($\models^{s}=\models^{d}$)
and Theorem~\ref{thm:rangehenkin} ($\models^{d}$ is the minimal embedding over
all interpretations); the full readings likewise agree,
$\models^{d'}=\models^{s'}$.}\label{tab:classic}
\setlength{\tabcolsep}{5pt}
\begin{tabular}{lccccc}
\hline
Schema (and witness, where valid) & $\models^{d}$ & $\models^{d'}$ & $\models^{s}$ & $\models^{s'}$ & $\models^{m}$\\
\hline
Complement\quad $\lambda d.\,\neg S(d)$ & $\times$ & $\checkmark$ & $\times$ & $\checkmark$ & $\times$\\
Intersection\quad $\lambda d.\,S_1(d)\wedge S_2(d)$ & $\times$ & $\checkmark$ & $\times$ & $\checkmark$ & $\times$\\
Union\quad $\lambda d.\,S_1(d)\vee S_2(d)$ & -- & $\checkmark$ & -- & $\checkmark$ & --\\
Separation\quad $\lambda d.\,S(d)\wedge I(P)(d,d)$ & -- & $\checkmark$ & -- & $\checkmark$ & --\\
$P$-image\quad $\lambda d.\,\exists e.\,S(e)\wedge I(P)(e,d)$ & -- & $\checkmark$ & -- & $\checkmark$ & --\\
$P$-preimage\quad $\lambda d.\,\exists e.\,I(P)(d,e)\wedge S(e)$ & -- & $\checkmark$ & -- & $\checkmark$ & --\\
Reachability $x\!\rightsquigarrow\!y$ & $\times$ & $\times$ & $\times$ & $\times$ & $\times$\\
Reflexive reachability $x\!\rightsquigarrow\!x$ & $\checkmark$ & $\checkmark$ & $\checkmark$ & $\checkmark$ & $\checkmark$\\
$2$-colorability & $\times$ & $\times$ & $\times$ & $\times$ & $\times$\\
\hline
\end{tabular}
\end{table}

\paragraph{Boolean closure and graph operations.}
The first six rows are instances of the comprehension schema above: each asserts
the set defined by a formula~$\varphi$, i.e.\ $\exists X.\,\forall d.\,(X(d)\leftrightarrow\varphi(d))$,
with witness $\lambda d.\,\varphi(d)$ (the predicate shown in the table). Over the full
second-order domain ($\models^{d'}/\models^{s'}$, $E=\mathcal P(D)$) each witness is an ordinary HOL predicate supplied by a single \texttt{exI}, for
the Boolean cases~\cite{Buchi1960,Thomas1997} and the monadic graph
operations~\cite{Courcelle2012} alike (see the table).
After unfolding \texttt{DefD} (resp.\ \texttt{DefS}) and $\beta$-reducing, the deep and maximal-shallow goals agree up to the assignment-update
definitions, so the two proof scripts are structurally identical --- the d/s
alignment made visible on non-trivial second-order content. In the minimal embedding $\exists^{m}_{2}$ ranges over
$V_2=\mathit{nat}$ through $GG$, so the required predicate need not lie in
$\mathrm{Range}\,GG$ and \texttt{nitpick} reports countermodels to \emph{complement} and
\emph{intersection} (the other four fail alike); as both binders
range over infinite namespaces, \texttt{nitpick} certifies these only as
\emph{potential} countermodels, the separation $\models^d\subsetneq\models^{d'}$ itself being mechanised via
the diagonal instance (\textsf{Standard\_strictly\_stronger}, \S\ref{sec:surj}).

\paragraph{Reachability and colorability.}
Reachability and $2$-colorability are invalid in MSO independently of the
standard/general choice, and fail in \emph{all} columns: reachability
\cite{BasinKlarlund1995} --- ``if $Z$ contains $x$ and is closed under
$P$-successors then $Z$ contains $y$'' --- has a \texttt{nitpick} countermodel
with $x,y$ $P$-disconnected (in $\models^{m}$ inferred via \textsf{Val} and Theorem~\ref{thm:rangehenkin}); its reflexive variant (conclusion $Z(x)$) follows from the first conjunct alone and needs no witness, so it survives even
the minimal embedding; $2$-colorability~\cite{Thomas1997} fails on any graph with a
$P$-self-loop (or, classically, an odd cycle such as $K_3$). The pattern --- closure rows valid under the full readings
$\models^{d'}/\models^{s'}$ but refuted under the general
$\models^{d}/\models^{s}/\models^{m}$, reachability and
$2$-colorability uniformly refuted --- is the inclusion $\models^d\subsetneq\models^{d'}$ of \S\ref{sec:surj}: among the closure rows, the minimal embedding refutes exactly the principles
separating the standard from the general reading; reflexive reachability holds
throughout.

\paragraph{Locale-based transfer.}
A short demonstration (theory
\texttt{MSOinHOL\_\allowbreak experiments\_\allowbreak locale.thy}%
\extonly{, Fig.~\ref{fig:MSOinHOL_experiments_locale.thy}}) proves a schema purely from
minimal-shallow definitions, inside the locale, and transfers it to the
range-relative deep embedding by a single application of
\texttt{MinS\_to\_Deep}, the directional reading of \textsf{FaithfulMS\_all}.

%% =================================================================
\section{Conclusion}\label{sec:concl}
%% =================================================================

\textsf{MSOinHOL} demonstrates that the deep-and-shallow methodology
of~\cite{C98} scales to a setting with \emph{two} binders over
different sorts. The new ingredients are a two-sorted substitution apparatus,
in which each binder is transparent for the other namespace, and a
locale-based minimal embedding with two HOL binders and an all-interpretations
faithfulness theorem. The
powerset character of the second-order domain turns the surjectivity problem into
a clean dichotomy. On one side, range-relative validity coincides with the general
(Henkin-style) reading of MSO --- a two-sorted L\"owenheim--Skolem theorem
(\S\ref{sec:surj}). On the other, the standard reading validates strictly more
formulas, yet it too is captured exactly by the minimal embedding once one restricts to
\emph{elementary} interpretations (\textsf{Deep'\_to\_MinS}). The two readings thus
share one translation and differ only in the class of interpretations quantified over. A suite of classical MSO landmarks turns this dichotomy into a table of validity proofs, \texttt{nitpick}-found
countermodels, and entailed cases.

The ingredients are individually classical: downward L\"owenheim--Skolem has been
formalised in Isabelle for first-order logic~\cite{FOL-Fitting}; the distinction
between the \emph{standard} and the \emph{general} (Henkin-style) reading of
second-order semantics is textbook~\cite{Vaananen2011}; and the general reading is known to be a two-sorted first-order
theory~\cite{Manzano1996}. Mechanised MSO in Isabelle has so far targeted decision procedures for finite
words and WS1S~\cite{TraytelNipkow2015,Traytel2015}; our focus, the semantics of
full and general MSO, is complementary. What is new, to our knowledge, is carrying a mechanised downward
L\"owenheim--Skolem theorem through to \emph{monadic second-order} validity. The construction is a two-sorted elementary-substructure argument that takes
each set it needs as a witness from the model's own admissible collection~$E$
rather than forming it as an arbitrary subset, which keeps it in the general
class.

\paragraph{Use in teaching.}
The development also doubles as teaching material for graduate and even
undergraduate logic courses: because the shallow embeddings turn MSO formulae
into ordinary Isabelle/HOL goals, students can explore the difference between
the standard and the general reading hands-on, with \texttt{sledgehammer} and
\texttt{nitpick} --- a proof in one column, a countermodel in another.

\paragraph{Availability.}
The full Isabelle/HOL development is published as the Archive of Formal Proofs
entry \textsf{MSOinHOL}~\cite{MSOinHOL-AFP}: seventeen theory files (ten core, seven supplementary), loaded in dependency
order by the session \texttt{ROOT}; the published entry adds one further theory, \texttt{experiments\_extra}, not
needed here. The development has been checked in Isabelle2025-2. \procorext{}{An extended version with the full L\"owenheim--Skolem proof and the
theory sources is on arXiv~\cite{msoinhol-ext}.}

\paragraph{Disclosure on the use of generative AI.}
The authors used generative AI assistants (Anthropic's Claude family of models)
to draft and shorten prose, to propose and shorten some proofs, and to maintain
cross-references between the LaTeX source and the Isabelle/HOL files. The
mathematical content, the theory development, and the design choices are the
authors'; all text and proofs in the final manuscript have been verified by the
authors, who take full responsibility for the content.

\bibliographystyle{splncs04}
\bibliography{extra,chris,literature,bibliography}

%% The two appendices are part of the extended (arXiv) version only; in the
%% proceedings version (\extendedfalse) the body reroutes all references to them.
\ifextended
\clearpage
\appendix
\section{The converse of \textsf{FaithfulMS\_all}: range vs.\ general (Henkin-style) validity}\label{app:proof}

We prove Theorem~\ref{thm:rangehenkin}. Write $\models^d_r\varphi$ for
range-relative validity,
\[
  \models^d_r\varphi \ \equiv\ \forall I\,g\,G.\
  \langle I,\mathrm{Range}\,g,\mathrm{Range}\,G\rangle,g,G\models^d\varphi .
\]
The forward direction $\models^d\varphi\Rightarrow\;\models^d_r\varphi$ is immediate
(it is \textsf{Deep\_to\_MinS}: $g$ maps into $\mathrm{Range}\,g$ and $G$ into
$\mathrm{Range}\,G$, so both are admissible). We prove the converse.

\medskip\noindent\textbf{Setup.}
A model $\langle I,D,E\rangle$ is read as a two-sorted first-order structure over
the fixed individual type: sort~$1$ is $\{d : D\,d\}$, sort~$2$ is $\{S : E\,S\}$,
each relation symbol $r$ denotes the binary relation $I\,r$ on sort~$1$, and a
single membership relation is defined by $\mathrm{mem}(d,S)\Leftrightarrow S(d)$.
Under this reading $\langle I,D,E\rangle,g,G\models^d\varphi$ is ordinary
two-sorted FO satisfaction, with $\exists^d x$ ranging over sort~$1$ and
$\exists^d_2 X$ over sort~$2$; \emph{no} closure on $E$ is imposed.

\medskip\noindent\textbf{Coincidence (generalising \texttt{L12}/\texttt{N12}).}
By induction on $\varphi$: if $g,g'$ agree on every first-order variable free in
$\varphi$ and $G,G'$ agree on every second-order variable free in $\varphi$, then
$\langle I,D,E\rangle,g,G\models^d\varphi \Leftrightarrow
\langle I,D,E\rangle,g',G'\models^d\varphi$. The binder cases use that an update
$g[y\!\leftarrow\!d]$ (resp.\ $G\langle Y\!\leftarrow\!S\rangle$) preserves
agreement on the free variables of the body.

\begin{proof}
Converse of Theorem~\ref{thm:rangehenkin}, contrapositively. Suppose
$\not\models^d\varphi$. Then there are $I,D,E,g,G$ with $g$ into $D$, $G$ into
$E$, and $\langle I,D,E\rangle,g,G\not\models^d\varphi$.

\emph{Step~1 (downward L\"owenheim--Skolem).} Build sub-predicates
$D_0\sqsubseteq D$ and $E_0\sqsubseteq E$ by closing
$\mathrm{Range}\,g \cup \mathrm{Range}\,G$ under Tarski--Vaught witnesses for the
existential subformulae of $\varphi$: whenever a sort-$1$ existential has a
witness in $D$ for a parameter tuple drawn from the current stage, add one to
$D_0$; whenever a sort-$2$ existential has a witness $S\in E$, add that $S$ to
$E_0$. Iterating $\omega$ times yields $D_0,E_0$ that are \emph{countable} (the
language and the two ranges are countable, and each stage adds only countably
many witnesses) and elementary for $\varphi$. By the Tarski--Vaught test,
$\langle I,D_0,E_0\rangle,g,G\models^d\psi \Leftrightarrow
\langle I,D,E\rangle,g,G\models^d\psi$ for every subformula $\psi$ of $\varphi$;
in particular $\langle I,D_0,E_0\rangle,g,G\not\models^d\varphi$. The
set-witnesses are honest elements of $E$, so $E_0\sqsubseteq E$: \emph{no
comprehension is used}.

\emph{Step~2 (reindexing).} $D_0$ and $E_0$ are countable and nonempty (they
contain the nonempty ranges $\mathrm{Range}\,g$, $\mathrm{Range}\,G$). Since
$V=V_2=\mathit{nat}$ are infinite and only finitely many variables occur free in
$\varphi$, choose $g'$ agreeing with $g$ on the free first-order variables of
$\varphi$ with $\mathrm{Range}\,g'=D_0$, and $G'$ agreeing with $G$ on the free
second-order variables with $\mathrm{Range}\,G'=E_0$ --- fix the finitely many
free slots, and let the cofinitely many remaining slots enumerate the countable
target.

\emph{Step~3 (conclude).} By Coincidence,
$\langle I,D_0,E_0\rangle,g',G'\not\models^d\varphi$. But $D_0=\mathrm{Range}\,g'$
and $E_0=\mathrm{Range}\,G'$, hence
$\langle I,\mathrm{Range}\,g',\mathrm{Range}\,G'\rangle,g',G'\not\models^d\varphi$,
i.e.\ $\not\models^d_r\varphi$, as required.
\end{proof}

\medskip\noindent\textbf{Remark (the standard reading, exactly).}
The hull above is not only the engine of the converse; packaged with its
truth-preservation it is a self-contained downward L\"owenheim--Skolem theorem
\textsf{DownwardLowenheimSkolem} --- every countable nonempty seed sub-pair extends to a
countable \emph{elementary} substructure $\langle I,N,M\rangle\subseteq_{\!E}\langle I,D,E\rangle$.
Applied to the full model and combined with \textsf{FaithfulMD} and the surjective reindexing, it
yields the strongest faithfulness theorem \textsf{Deep'\_to\_MinS}:
$\models^{d'}\varphi$ iff $\models^m(\!\lvert\varphi\rvert\!)$ holds for every minimal interpretation
whose range model is an elementary substructure of the full model (locale \textsf{MinS\_ES\_Univ}).
The general reading is thus captured by the minimal embedding over \emph{all} interpretations, the standard reading over the
\emph{elementary} ones; the gap between them is witnessed by the comprehension instance below.

\medskip\noindent\textbf{Remark (why the standard reading is excluded by L\"owenheim--Skolem).}
The construction stays within the general class precisely because the sort-$2$
witnesses are taken from the given $E$. It cannot reach the standard reading
$\models^{d'}$ ($E=\mathcal P(D)$) by an \emph{all-interpretations} argument: by
\textsf{comprehension\_atom} the formula
$\exists X.\,\forall x.\,(X(x)\leftrightarrow P(x,x))$ is $\models^{d'}$-valid,
yet a general model whose $E$ omits the diagonal $\{d : P(d,d)\}$ refutes it, so
it is not $\models^d$-valid. Thus $\models^d\subsetneq\models^{d'}$, and no
range-relative L\"owenheim--Skolem argument can bridge the two; the elementary restriction of
\textsf{Deep'\_to\_MinS} cannot be relaxed to all interpretations. A structured Isabelle
development of the argument above is provided as
\texttt{MSOinHOL\_\allowbreak lowenheim\_\allowbreak skolem.thy}: the coincidence
and reindexing lemmas, the two-sorted hull \textsf{ls\_hull} (with its countable
Tarski--Vaught construction \textsf{skolem\_hull} and the truth preservation
\textsf{truth\_pres} it yields), the assembled converse, the elementary-substructure relation
$\subseteq_{\!E}$ with \textsf{DownwardLowenheimSkolem} and \textsf{Deep'\_to\_MinS}, and this
non-reduction remark are written out in full and machine-checked under Isabelle2025-2.

\section{Rendered Isabelle/HOL sources}\label{app:src}
For completeness every theory file of the development is rendered here as a
syntax-highlighted source figure (the five theories already presented in the
body --- the deep embedding, the maximal- and minimal-shallow embeddings,
faithfulness, and the experiments --- are not repeated here).
Long theories are split into consecutive parts of at most eighty lines; the
line-number gutter is continuous, so the parts read as one listing. The full
session (file \texttt{ROOT}) loads all theories in dependency order.

\thyfigH{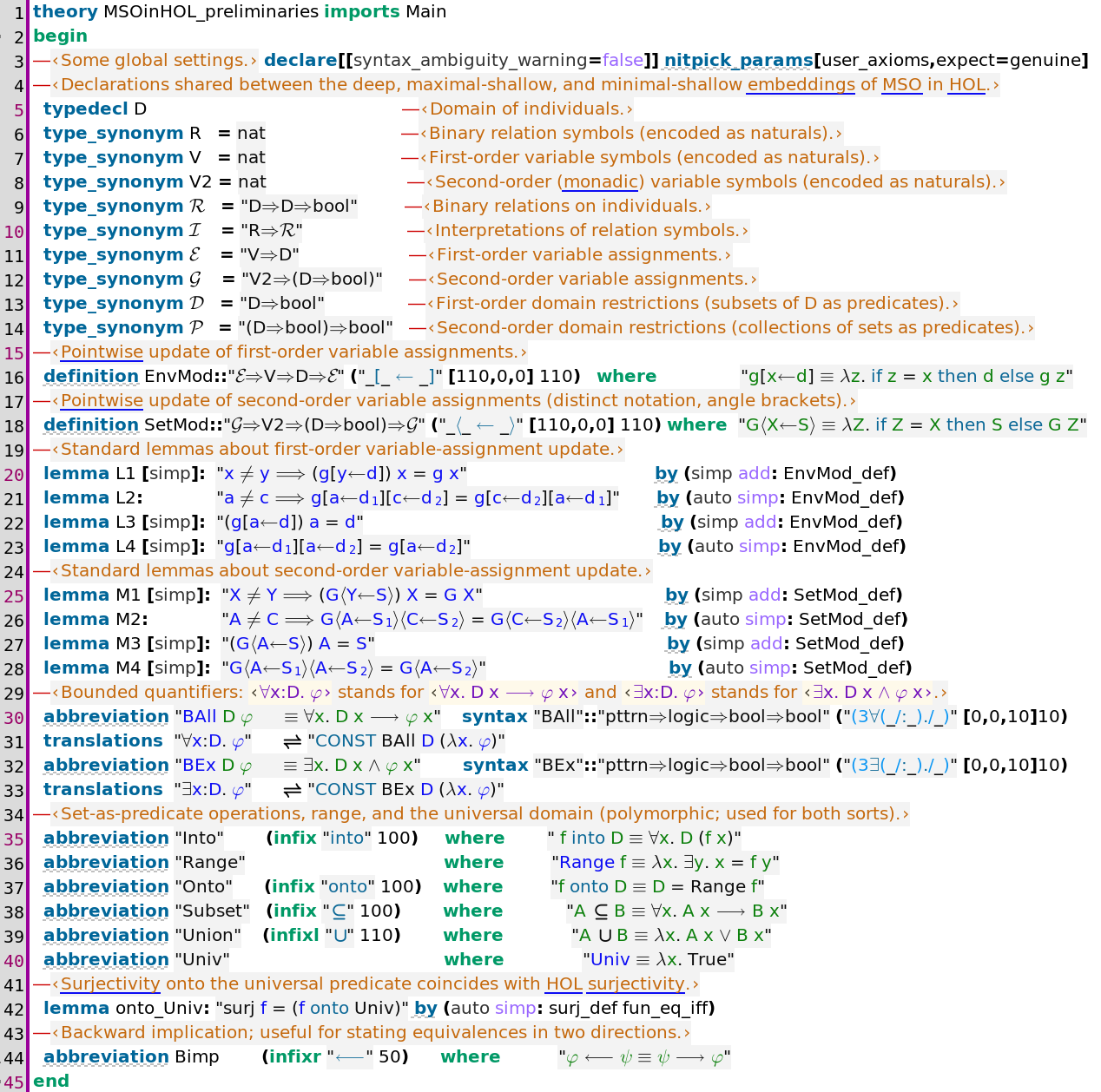}{Theory \texttt{MSOinHOL\_preliminaries.thy}: the
shared type declarations, the two variable-assignment update operators
$g[x{\leftarrow}d]$ and $G\langle X{\leftarrow}S\rangle$ with their companion
lemmata, the bounded quantifiers, and the polymorphic set-as-predicate
operations.}

\thyfigH{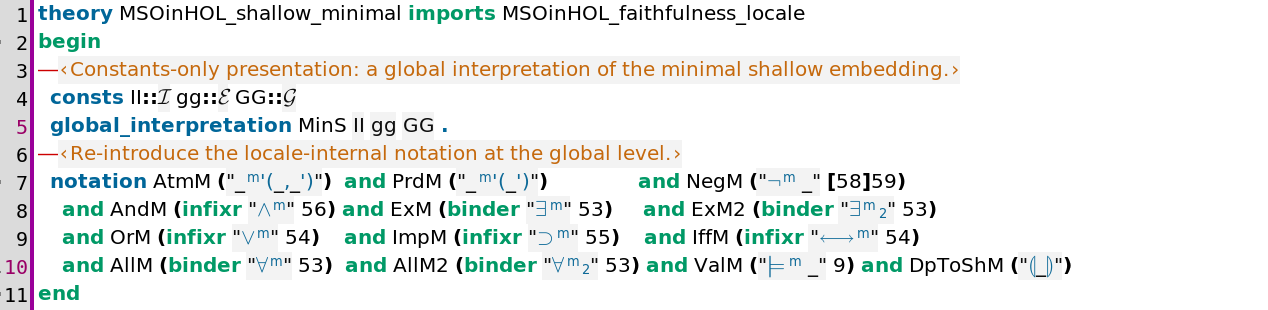}{Theory \texttt{MSOinHOL\_shallow\_minimal.thy}:
the constants-only \texttt{global\_interpretation} of \texttt{MinS} that re-issues
the locale notation at the level of three uninterpreted constants
\texttt{II},\texttt{gg},\texttt{GG}.}

\thyfigH{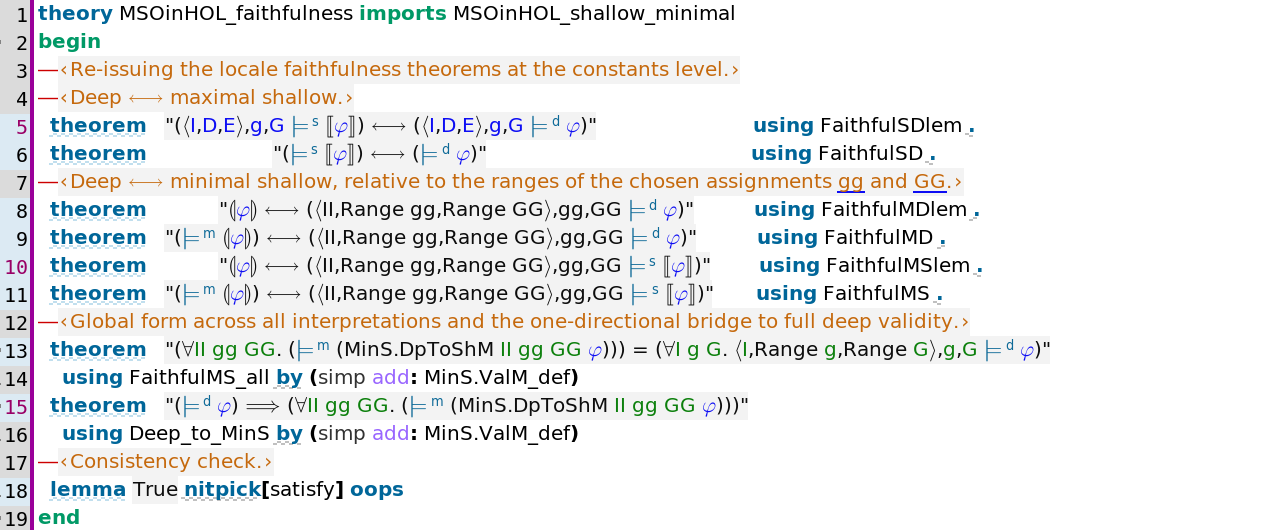}{Theory \texttt{MSOinHOL\_faithfulness.thy}: the
locale faithfulness theorems re-issued at the constants level (deep$\leftrightarrow$maximal-
and deep$\leftrightarrow$minimal-shallow correspondences, the all-interpretations theorem, and the
one-directional bridge to deep validity).}

\thyfigpartH{MSOinHOL_deep_subst_lemma}{1}{Theory
\texttt{MSOinHOL\_deep\_subst\_lemma.thy} (Part~1 of 3): Part~A (first-order
variables) --- \texttt{is\_free}, \texttt{is\_bound}, \texttt{fresh} with their
consequence lemmata, the irrelevance lemma \texttt{L12}, variable-for-variable
substitution, and the size-based induction principles
\texttt{SInduct}/\texttt{QInduct}.}

\thyfigpartH{MSOinHOL_deep_subst_lemma}{2}{Theory
\texttt{MSOinHOL\_deep\_subst\_lemma.thy} (Part~2 of 3): the substitutability
predicate, \textsf{SubstitutionLemma}, alphabetic renaming
\texttt{ren\_for\_subst}, safe substitution \texttt{ren\_subst}, and
\texttt{L17}--\texttt{L29} (Part~A concluded); then the start of Part~B
(second-order variables) --- \texttt{is\_free2}, \texttt{is\_bound2},
\texttt{fresh2}, and \texttt{N5}--\texttt{N11}.}

\thyfigpartH{MSOinHOL_deep_subst_lemma}{3}{Theory
\texttt{MSOinHOL\_deep\_subst\_lemma.thy} (Part~3 of 3): Part~B (second-order
variables) concluded --- the irrelevance lemma
\texttt{N12}, second-order substitution \texttt{Subst2},
\textsf{SubstitutionLemma2}, renaming
\texttt{ren\_for\_subst2}/\texttt{ren\_subst2}, and \texttt{N17}--\texttt{N29}.}

\thyfigH{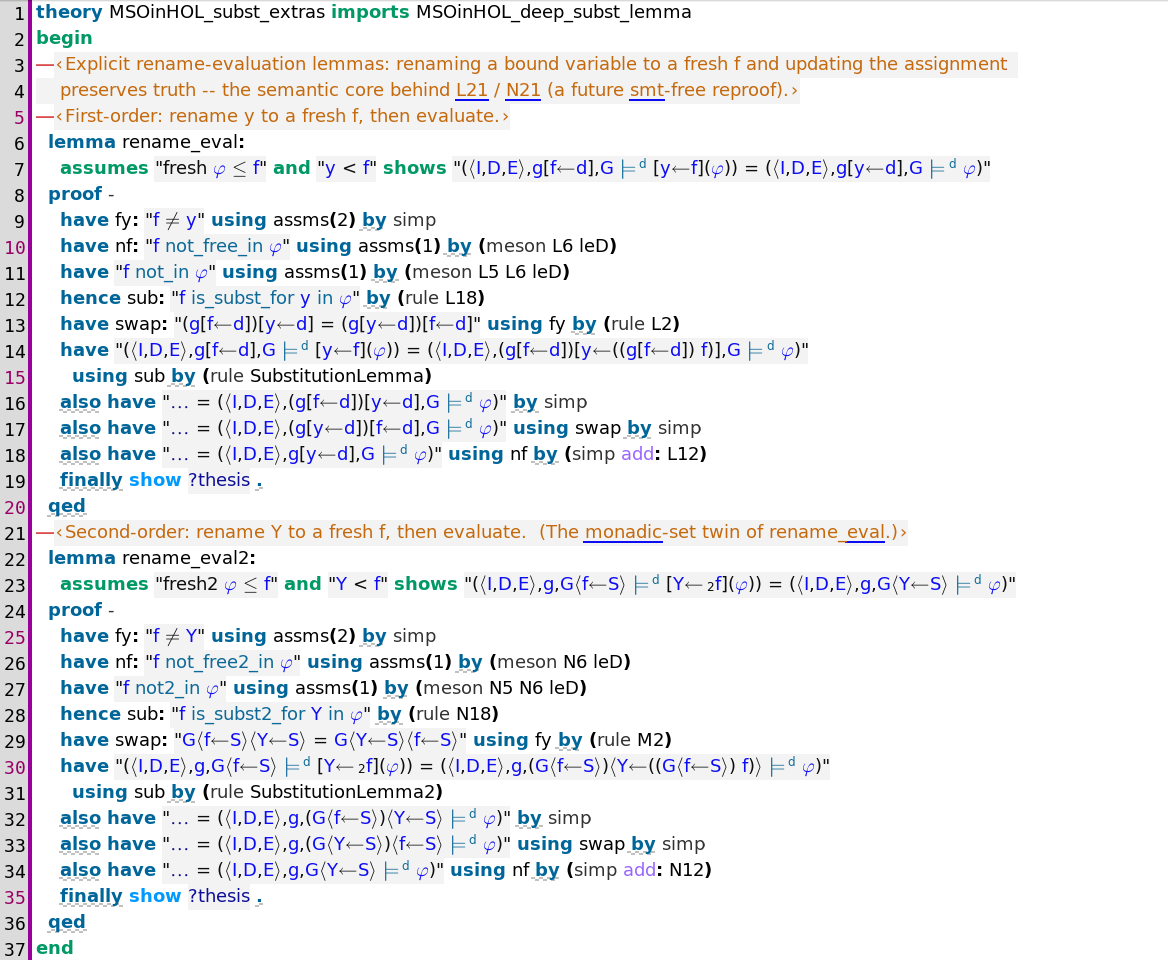}{Theory \texttt{MSOinHOL\_subst\_extras.thy}: the
explicit rename-evaluation lemmas \textsf{rename\_eval} (first-order) and
\textsf{rename\_eval2} (second-order).}

\thyfigH{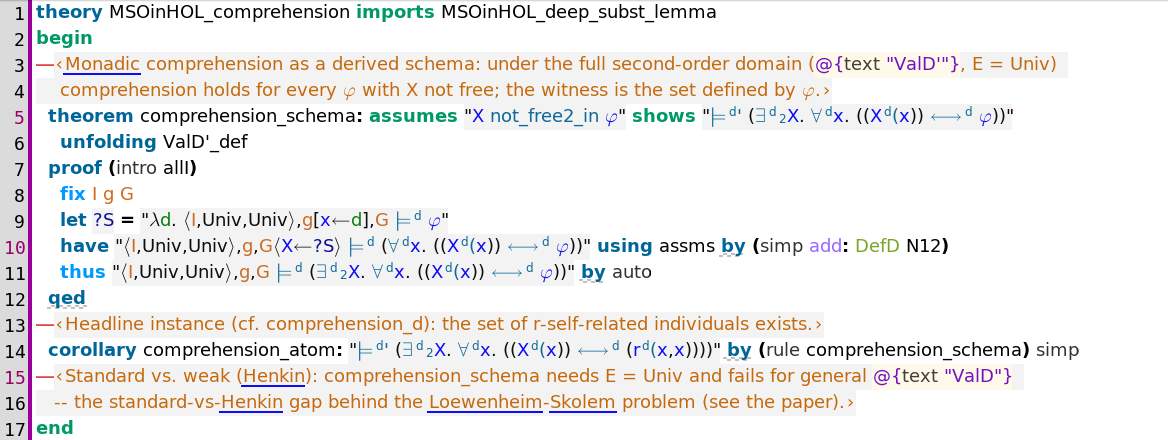}{Theory \texttt{MSOinHOL\_comprehension.thy}: the
general monadic \textsf{comprehension\_schema} over the full second-order domain
with its witness, the corollary \textsf{comprehension\_atom}, and the discussion
relating comprehension to standard vs.\ general (Henkin-style) semantics.}

\thyfigH{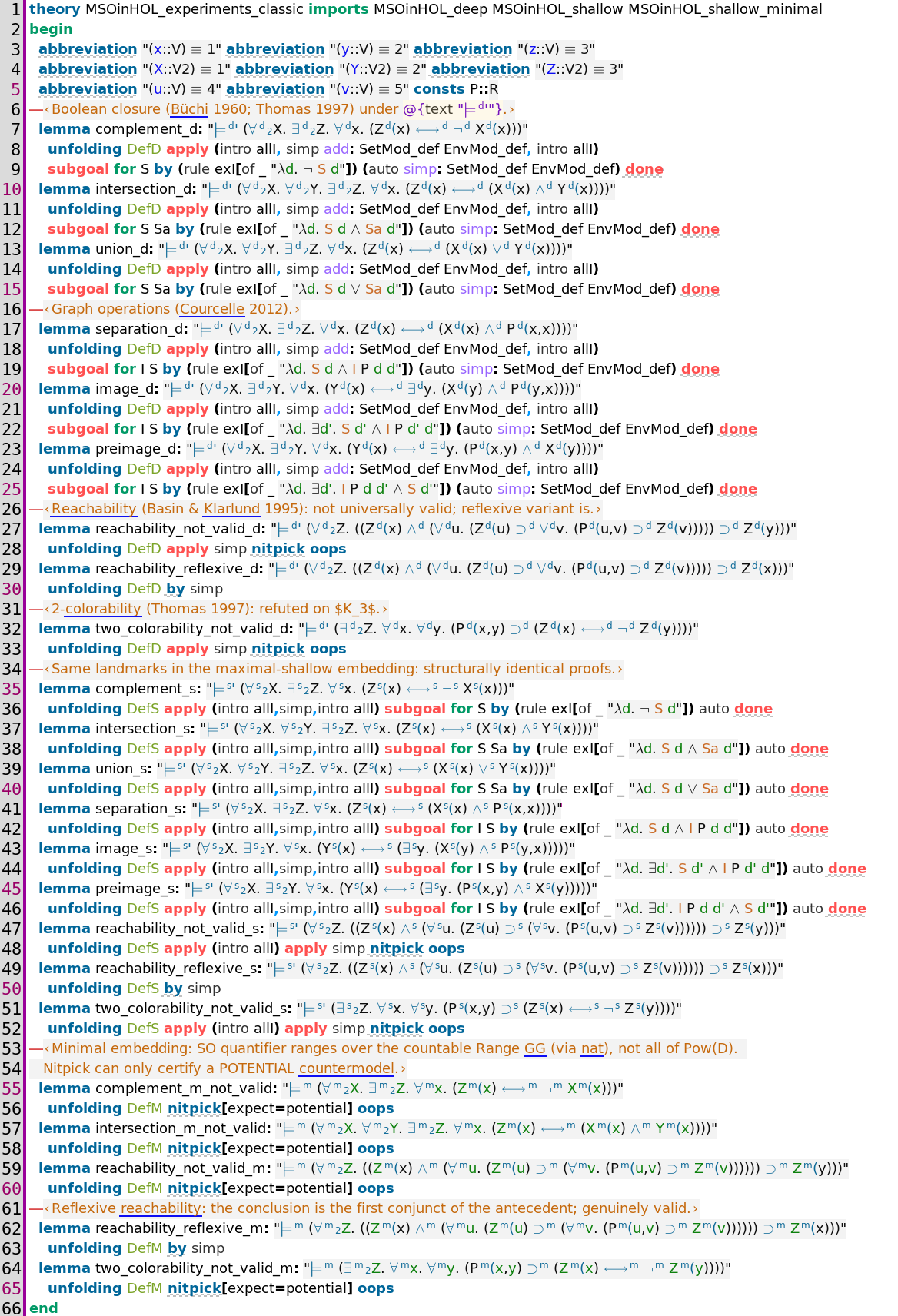}{Theory
\texttt{MSOinHOL\_experiments\_classic.thy}: the classical MSO landmarks in the
deep and maximal-shallow embeddings --- the six closure schemata (complement,
intersection, union, separation, $P$-image, $P$-preimage), each with an explicit
second-order witness via \texttt{exI}, and reachability and $2$-colorability,
refuted by \texttt{nitpick}; and the minimal embedding, where \texttt{nitpick} reports (potential) countermodels
to Boolean closure, reachability, and $2$-colorability while reflexive reachability
remains valid.}

\thyfigH{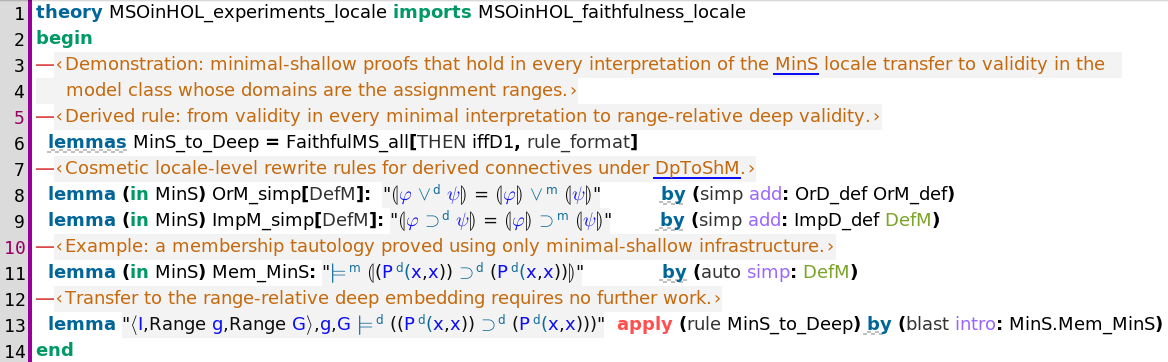}{Theory
\texttt{MSOinHOL\_experiments\_locale.thy}: locale-based transfer to the deep
embedding --- the directional reading \texttt{MinS\_to\_Deep} of
\textsf{FaithfulMS\_all}, two cosmetic rewrite rules, a schema proved inside the
locale, and its one-line transfer.}

\thyfigpartH{MSOinHOL_lowenheim_skolem_lemmas}{1}{Theory
\texttt{MSOinHOL\_lowenheim\_skolem\_lemmas.thy} (Part~1 of 4): the supporting
machinery for the L\"owenheim--Skolem construction --- the coincidence lemma
\textsf{coincidence} (a two-sorted generalisation of \texttt{L12}/\texttt{N12}),
the padding maps \textsf{padA}/\textsf{padB}, the Skolem-witness operators
\textsf{wF}/\textsf{wS} with their membership lemmas, the $\omega$-stage
operator \textsf{stage}, the stage-introduction lemmas
\textsf{fst\_stage\_Suc\_intro}/\textsf{snd\_stage\_Suc\_intro}, and the
padding-truth lemmas \textsf{pad\_truth\_eq\_FO}/\textsf{pad\_truth\_eq\_SO}.}

\thyfigpartH{MSOinHOL_lowenheim_skolem_lemmas}{2}{Theory
\texttt{MSOinHOL\_lowenheim\_skolem\_lemmas.thy} (Part~2 of 4): the
witness-realisation lemmas
\textsf{wF\_realizes}/\textsf{wS\_realizes}, the per-step Skolem closure
\textsf{skolem\_step\_FO}/\textsf{skolem\_step\_SO}, the subset bookkeeping
\textsf{fst\_stage\_Suc\_subset}/\textsf{snd\_stage\_Suc\_subset}, the stage
monotonicity \textsf{fst\_stage\_mono}/\textsf{snd\_stage\_mono}, and the
common-stage lemma \textsf{common\_stage\_lifted}.}

\thyfigpartH{MSOinHOL_lowenheim_skolem_lemmas}{3}{Theory
\texttt{MSOinHOL\_lowenheim\_skolem\_lemmas.thy} (Part~3 of 4): the
$\omega$-union membership lemmas
\textsf{wF\_in\_omega\_union}/\textsf{wS\_in\_omega\_union}, the first- and
second-order Tarski--Vaught closures
\textsf{stage\_TV\_FO}/\textsf{stage\_TV\_SO}, the inductive countability step
\textsf{cntDE\_step}, the carriers-stay-inside lemmas
\textsf{fst\_stage\_subD}/\textsf{snd\_stage\_subE}, stage-wise and
$\omega$-countability \textsf{stage\_countable}/\textsf{stage\_omega\_countable},
and the surjective reindexing \textsf{reindex\_one}.}

\thyfigpartH{MSOinHOL_lowenheim_skolem_lemmas}{4}{Theory
\texttt{MSOinHOL\_lowenheim\_skolem\_lemmas.thy} (Part~4 of 4): the conclusion of
\textsf{reindex\_one}, the surjective reindexing \textsf{reindex}, and the
combined \textsf{reindex\_coincide} lemma.}

\thyfigpartH{MSOinHOL_lowenheim_skolem}{1}{Theory
\texttt{MSOinHOL\_lowenheim\_skolem.thy} (Part~1 of 2): the elementary-substructure
layer --- the relation \textsf{ElementarySubstructure} ($\subseteq_{\!E}$) and the locales
\textsf{MinS\_ES}/\textsf{MinS\_ES\_Univ} --- range-relative validity \textsf{RangeValid},
the easy direction \textsf{ValD\_imp\_RangeValid}, truth preservation \textsf{truth\_pres}
for a Tarski--Vaught-closed sub-pair, the countable seed hull \textsf{skolem\_hull},
and the first-class downward L\"owenheim--Skolem theorem \textsf{DownwardLowenheimSkolem}.}

\thyfigpartH{MSOinHOL_lowenheim_skolem}{2}{Theory
\texttt{MSOinHOL\_lowenheim\_skolem.thy} (Part~2 of 2): the strongest faithfulness theorem
\textsf{Deep'\_to\_MinS} (standard validity $\Leftrightarrow$ minimal validity over elementary
interpretations) with its corollary \textsf{Faithful\_to\_Standard}; then the earlier
general-reading route, derived in one step from \textsf{DownwardLowenheimSkolem} ---
the assembled elementary-substructure hull \textsf{ls\_hull}, the range reduction
\textsf{Range\_reduction}, the main theorem \textsf{RangeValid\_imp\_ValD} with
corollaries \textsf{RangeValid\_iff\_ValD}/\textsf{Faithful\_to\_Henkin}, and the
negative result \textsf{Standard\_strictly\_stronger} with its explicit diagonal
countermodel.}

\thyfigH{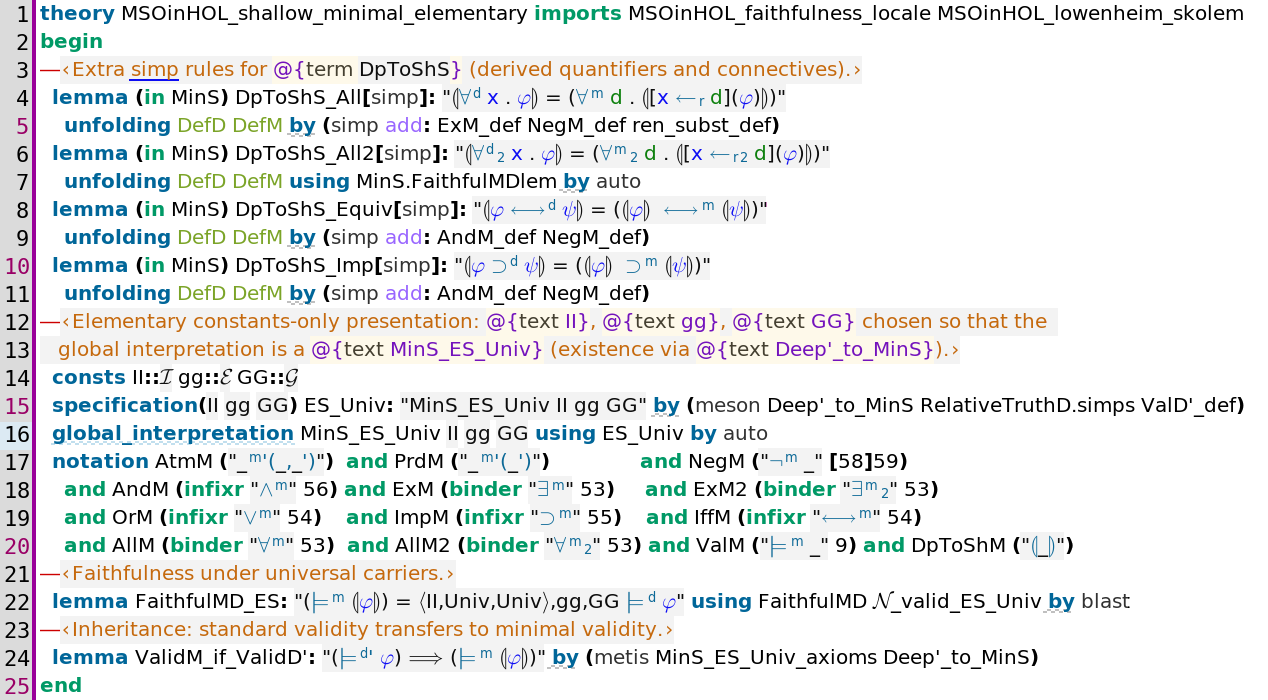}{Theory
\texttt{MSOinHOL\_shallow\_minimal\_elementary.thy}: an \emph{elementary}
constants-only presentation of the minimal shallow embedding --- four simplification lemmas for the translation $(\!\lvert\cdot\rvert\!)$ on the
derived quantifiers and connectives (named \texttt{DpToShS\_All},
\texttt{DpToShS\_All2}, \texttt{DpToShS\_Equiv}, \texttt{DpToShS\_Imp}, although
they concern \texttt{DpToShM}), the three uninterpreted constants
\texttt{II}, \texttt{gg}, \texttt{GG}, a \texttt{specification} establishing the
existence of a global interpretation that forms a \texttt{MinS\_ES\_Univ}
elementary substructure (with existence justified by
\textsf{Deep'\_to\_MinS}), the corresponding \texttt{global\_interpretation}
re-issuing the locale notation at the constants level, and the inheritance
lemmas \texttt{FaithfulMD\_ES} (faithfulness under universal carriers) and
\texttt{ValidM\_if\_ValidD'} (standard validity lifts to minimal validity).}

\thyfigpartH{MSOinHOL_experiments_classic_elementary}{1}{Theory
\texttt{MSOinHOL\_experiments\_classic\_elementary.thy} (Part~1 of 2): additional
simplification rules for the derived second-order quantifier and equivalence
(\texttt{ren\_for\_subst2\_simp\_All}, \texttt{subst2\_all},
\texttt{free2\_in\_equiv}, \texttt{free2\_in\_all}, \texttt{free2\_in\_all2}),
the convenient variable abbreviations, and the full set of deep-embedding
classical MSO landmarks --- Boolean closure (\texttt{complement\_d},
\texttt{intersection\_d}/\texttt{intersection\_d'}, \texttt{union\_d}),
monadic graph operations (\texttt{separation\_d}, \texttt{image\_d},
\texttt{preimage\_d}), reachability (\texttt{reachability\_not\_valid\_d},
\texttt{reachability\_reflexive\_d}), and $2$-colorability
(\texttt{two\_colorability\_not\_valid\_d}).}

\thyfigpartH{MSOinHOL_experiments_classic_elementary}{2}{Theory
\texttt{MSOinHOL\_experiments\_classic\_elementary.thy} (Part~2 of 2): the
same landmarks in the maximal-shallow embedding (the deep$\leftrightarrow$maximal-shallow scripts are structurally identical,
differing only in the unfolded assignment-update definitions), and their lifts to the elementary global
interpretation --- the $(\!\lvert\cdot\rvert\!)$-normal-form lemmas (\texttt{compl\_m\_DpToSh},
\texttt{inters\_m\_DpToSh'}, \texttt{intersect\_m\_DpToSh}), the helper
combinators (\texttt{DpToSh\_alleqI}, \texttt{all\_m\_eqI}, \texttt{DpToSh\_simps}),
and the \texttt{ValidM\_if\_ValidD'}-style validity transfers
(\texttt{compl\_m\_valid}, \texttt{intersection\_m\_valid}); reachability
and $2$-colorability are recorded as open goals (\texttt{oops}): for the
underspecified elementary interpretation their failure is not determined, in
contrast to the general minimal embedding (the in-file comment ``remain
refuted'' predates this analysis).}

\fi % end of extended-only appendices

\end{document}